\documentclass[11pt,reqno]{amsart}
\usepackage[foot]{amsaddr}
\usepackage{amsmath, amssymb}
\usepackage{amsthm}
\usepackage[a4paper,margin=1.25in]{geometry}
\usepackage{fullpage} 
\usepackage[english]{babel}
\usepackage{hyperref}
\usepackage{multirow}

\usepackage{accents}

\usepackage{mathdots}

\usepackage{todonotes}

\usepackage{tikz}
\usepackage{tikz-cd}
\usetikzlibrary{calc}
\usepackage{color, mathtools,epsfig, graphicx}
\usetikzlibrary{positioning}

\newtheorem{theorem}{Theorem}[section]

\newtheorem{proposition}[theorem]{Proposition}

\newtheorem{corollary}[theorem]{Corollary}

\newtheorem{lemma}[theorem]{Lemma}

\newtheorem{notation}[theorem]{Notation}

\theoremstyle{definition}

\newtheorem{definition}[theorem]{Definition}

\theoremstyle{remark}

\newtheorem{remark}[theorem]{Remark}

\addto\extrasenglish{

}

\numberwithin{equation}{section}

\newcommand{\C}{\mathbb{C}}

\newcommand{\Z}{\mathbb{Z}}
\newcommand{\p}{\mathbb{P}}
\newcommand{\Pic}{\operatorname{Pic}}

\newcommand{\Div}{\operatorname{Div}}

\newcommand{\pain}[1]{\text{P}_{\mathrm{#1}}}

\begin{document}

\title[On the geometry of a 4-dimensional extension of a $q$-Painlev\'e I equation with symmetry type $A_1^{(1)}$]{On the geometry of a 4-dimensional extension of a $q$-Painlev\'e I equation with symmetry type $A_1^{(1)}$}

\author[AS]{Alexander Stokes$^{1}$}

\address{$^1$Waseda Institute for Advanced Study, Waseda University, 1-21-1 Nishi Waseda Shinjuku-ku Tokyo 169-0051, Japan.}

\author[TT]{Tomoyuki Takenawa$^2$}

\address{$^2$Faculty of Marine Technology, Tokyo University of Marine Science and Technology, 2-1-6 Etchu-jima, Koto-Ku, Tokyo 135-8533, Japan.}

\author[ASC]{Adrian Stefan Carstea$^3$}
\address{$^3$Department of Theoretical Physics, Institute of Physics and Nuclear Engineering\\ Reactorului 15, 077125, Magurele, Bucharest, Romania}

\email{stokes@aoni.waseda.jp, 
takenawa@kaiyodai.ac.jp, 
carstea@gmail.com} 

\keywords{discrete Painlev\'e equation, birational map, pseudo-isomorphism, rational 4-fold, Cremona isometry, affine Weyl group}

\maketitle




\begin{abstract}
We present a geometric study of a four-dimensional integrable discrete dynamical system which extends the autonomous form of a $q$-Painlev\'e I equation with symmetry of type $A_1^{(1)}$. 
By resolution of singularities it is lifted to a pseudo-automorphism of a rational variety obtained from $({\mathbb P}^1)^{\times 4}$ by blowing up along 28 subvarieties and we use this to establish its integrability in terms of conserved quantities and degree growth. 
We embed this rational variety into a family which admits an action of the extended affine Weyl group $\widetilde{W}(A_1^{(1)})\times \widetilde{W}(A_1^{(1)})$ by pseudo-isomorphisms.
We use this to construct two 4-dimensional analogues of $q$-Painlev\'e equations, one of which is a deautonomisation of the original autonomous integrable map.
\end{abstract}

\section{Introduction}
Singularities play an essential role in the study of integrability. Not only the Painlev\'e property (which concerns the absence of movable critical singularities of solutions of differential equations), but also singularity confinement have proven to be useful integrability detectors for continuous and discrete dynamical systems. 
Sakai \cite{sakai} realised that singularity confinement is intimately related to the classical desingularisation (blowing-up/down) procedure in birational algebraic geometry and that the singularity confinement property of second-order discrete systems defined by birational mappings means they are turned into regular automorphisms of rational surfaces (or families of isomorphisms of families of surfaces in the non-autonomous case of discrete Painlev\'e equations). 
This was a crucial observation which gave rise to a geometric approach to the classification of all discrete Painlev\'e equations of order two in terms of generalised Halphen surfaces, which are smooth complex projective rational surfaces defined by the existence of a special effective anticanonical divisor. 
In this approach, one constructs families of surfaces obtained by blowing up nine (possibly infinitely near) points in the complex projective plane $\mathbb{P}^2(\C)$. 
The configurations of points can be generic or degenerate, and these determine different arrangements of irreducible components of an effective anticanonical divisor, encoded in a Dynkin diagram of affine type.
The orthogonal complement of the divisor classes of these irreducible components in the Picard lattice of the surface is the root lattice of an affine root system, and the associated extended affine Weyl group acts as symmetries of the family via Cremona isometries. 
Discrete Painlevé equations correspond to the translational elements of the group and they act as automorphisms of the family of surfaces. 
One of the most important results of Sakai’s approach is the discovery of the elliptic discrete Painlev\'e equations whose coefficients are written in terms of elliptic functions, these elliptic equations being the most general systems at the top of the classification.

Recently, the study of four-dimensional Painlev\'e systems in the continuous case has been proposed using isomonodromic deformation of linear equations \cite{37, 27, Kawakami-I, Kawakami-II,Kawakami-III}. 
Higher-dimensional analogues of discrete Painlev\'e equations have also been discovered through various approaches, including through reductions of discrete integrable hierarchies (e.g. \cite{Tsuda-2010, Suzuki-2015}) and via birational representations of affine Weyl groups  including those coming from the theory of cluster algebras \cite{Masuda-2015,Nakazono-2023,Nakazono-2024, Okubo-Suzuki-2020, Okubo-Suzuki-2022, Masuda-Okubo-Tsuda-2021}. 
There are also QRT-type maps in higher dimensions \cite{Tsuda-2004, Alonso-Suris-Wei}, deautonomisation of which can potentially lead to discrete Painlev\'e-type systems.
B\"acklund transformations of higher-order analogues of differential Painlev\'e equations are the ``oldest'' source of higher discrete Painlev\'e equations, together with discrete Painlev\'e hierarchies obtained via the Lax formalism \cite{Joshi-Cresswell, Joshi-Alrashdi}. 

However, a geometric approach to higher-dimensional discrete Painlev\'e systems analogous to that of Sakai has only been carried out for a few examples.
While a mechanism for associating rational varieties to higher-dimensional discrete Painlevé equations via a kind of desingularisation of birational mappings is known, this is more complicated than in the Sakai framework since it may require the blowing-up of subvarieties of codimension higher than two. 
Although some studies of higher-dimensional birational mappings related to symmetries of varieties or dynamical systems have been reported, most of them consider only the case where varieties are obtained from projective space by \emph{blowing up only points} \cite{Dolgachev-Ortland-1989, Takenawa-2004, Bedford-Kim-2004, Bedford-Kim-2008}. 
In particular, 3D analogues of elliptic discrete Painlev\'e equations were constructed from blow-ups of points in $\p^3$ \cite{Takenawa-2004}.
In addition to this, the geometry has been worked out for higher $q$-analogues of the Sasano system \cite{Tsuda-Takenawa-2009}, the 4D Fuji-Suzuki-Tsuda system \cite{Takenawa-2021}, and B\"acklund transformations of the 4D Garnier system \cite{Takenawa-2024}.
The central problem in the geometric approach to higher-dimensional analogues of discrete Painlev\'e equations at present is that it is not yet clear what the general definition of higher-dimensional analogues of generalised Halphen surfaces should be -- we need more examples to inform the development of a general theory. 

To this end, in this paper we establish the geometric picture for the autonomous 4D system defined by the birational mapping  
\begin{equation}
\label{Map1}
\begin{aligned}
    \varphi: \left(\p^1\right)^{\times 4} &\rightarrow \left(\p^1 \right)^{\times 4}, \\
    (q_1,p_1,q_2,p_2)&\mapsto(\bar{q}_1,\bar{p}_1,\bar{q}_2,\bar{p}_2),
\end{aligned}
\qquad 
\left\{\begin{array}{rcl}
\bar{q}_1&=&\frac{1+bq_2}{q_1q_2p_2},\\
\bar{p}_1&=&q_2, \\
\bar{q}_2&=&\frac{1+bq_1}{q_1p_1q_2},\\
\bar{p}_2&=&q_1, 
\end{array} \right.
\end{equation}
where $\p^1=\p^1(\C)$ is the complex projective line, $q_1,p_1,q_2,p_2$ are affine coordinates, and $b\in \C^*$ is a parameter.
The motivation for considering the mapping \eqref{Map1} is that two of the authors (ASC and TT) studied in \cite{stefan} a four-dimensional two-component version of the autonomous version of an additive-type discrete Painlev\'e I equation of affine Weyl group symmetry type $A_2^{(1)}$. 
This led to four-dimensional discrete Painlev\'e systems with symmetries of type $A_5^{(1)}$. 
Accordingly, it is natural to see what happens in the case of a \emph{multiplicative-type} mapping under the same kind of two-component extension, and we have chosen a simple example of the autonomous version of a $q$-Painlev\'e I equation with symmetry type $A_1^{(1)}$:
\begin{equation}\label{qp1}
x_{n+1}x_{n-1}=\frac{a+bx_n}{x_n^2},
\end{equation}
where $a,b \in \C^*$.
This can be recast as a birational mapping 
\begin{equation} \label{qp1mapping}
\begin{aligned}
    \left(\p^1\right)^{\times 2} &\rightarrow \left(\p^1\right)^{\times 2}, \\
    (x,y)&\mapsto(\bar{x},\bar{y}),
\end{aligned}
\qquad 
\left\{\begin{array}{rcl}
\bar{x}&=&\frac{a+b x}{x^2y},\\
\bar{y}&=&x, 
\end{array} \right.
\end{equation}
by setting $x_{n-1}= y$, $x_n=x = \bar{y}$ and $x_{n+1}=\bar{x}$, 
Performing a similar procedure as in \cite{stefan} to obtain a two-component version of \eqref{qp1mapping} leads to the mapping \eqref{Map1}.

\begin{remark}
    The non-autonomous version of equation \eqref{qp1} (obtained by replacing the constant $b$ with a multiplicative-affine function of $n$, i.e. $b q^n$ where $b,q\in \C$ with $0 < |\text{q}|<1$) is usually dubbed a $q$-Painlev\'e I equation. 
    This is due to it admitting a continuum limit to the first Painlev\'e equation. 
    Indeed, one can see immediately that setting $t=\epsilon n$, $x_n=1+\epsilon^2 w(t)$, $b=4q^n$, $a=-3$, and $q=1-\epsilon^5/4$ and taking formally the continuum limit as $\epsilon \to 0$ in \eqref{qp1} leads to the classical Painlev\'e I equation $\tfrac{d^2w}{dt^2}=6w^2+t.$
\end{remark} 

In this paper we construct a rational variety $X$ to which the mapping \eqref{Map1} lifts to a pseudo-automorphism, obtained from $\left(\p^1\right)^{\times 4}$ by blowing up along 28 subvarieties of codimensions 2 or 3. 
From this we establish the integrability of the mapping \eqref{Map1}, by finding conserved quantities via the anticanonical linear system of $X$ and showing that it exhibits quadratic degree growth by computing the action of the mapping on cohomology.
We then deform this $X$ into a family of varieties $X_{\boldsymbol{b}}$, where $\boldsymbol{b}$ is a tuple of parameters, which admits an action of an extended affine Weyl group $\widetilde{W}(A_1^{(1)})\times \widetilde{W}(A_1^{(1)})$ by pseudo-automorphisms of the family, i.e. for each $w \in \widetilde{W}(A_1^{(1)}) \times \widetilde{W}(A_1^{(1)}) $ we have a pseudo-automorphism $w : X_{\boldsymbol{b}}\rightarrow X_{w(\boldsymbol{b})}$, where $\boldsymbol{b}\mapsto w(\boldsymbol{b})$ is an appropriate action on parameters.  
We use this to obtain two four-dimensional analogues of discrete Painlev\'e equations, one of which is a deautonomisation of the original autonomous mapping and the other corresponds to a different, non-conjugate, symmetry.

Note that if we extend equation \eqref{qp1} in a more straightforward manner, i.e. just by taking the trivial coupling of two copies of it, as
\begin{equation}
\label{Map1_alt}
\left\{
\begin{array}{rcl}
\bar{q}_1 &=& \frac{1+b q_1}{q_1^2 p_1},\\
\bar{p}_1 &=& q_1, \\[6pt]
\bar{q}_2 &=& \frac{1+b q_2}{q_2^2 p_2},\\
\bar{p}_2 &=& q_2,
\end{array}
\right.
\end{equation}
instead of \eqref{Map1}, the resulting system also leads to a family of rational varieties (essentially the direct product of the rational surfaces from the $q$-Painlev\'e I equation \eqref{qp1}) which admits the symmetry $\widetilde{W}(A_1^{(1)}) \times \widetilde{W}(A_1^{(1)})$. 
However, the construction of \eqref{Map1} is more intricate, making its symmetry much less apparent compared to that of \eqref{Map1_alt} from the trivial coupling.

\section{Blow-ups, pseudo-automorphisms and algebraic stability} \label{sec:blow-ups,pseudo-automorphisms}

We begin by recalling some definitions and facts about dynamical systems defined by birational mappings. 
For the sake of brevity we will include only those required to state our results and keep discussion to a minimum, but refer the reader to \cite{stefan, Dolgachev-Ortland-1989, Tsuda-Takenawa-2009} for more details and examples.
We work over $\C$ and all varieties will be smooth, projective and irreducible. 

Let us consider a rational map $f:\p^N \rightarrow \p^N$, given by a tuple of $N+1$ homogeneous polynomials having the same degree and without common polynomial factor: $f : {\bf x} \mapsto f({\bf x})=[ f_0({\bf x}), \dots ,f_{N+1}({\bf x}) ] \in \p^N$. 
The \emph{indeterminacy set} of $f$ is given by
\begin{equation*}
I(f)=\{{\bf x}\in\p^n~|~f_0({\bf x})=\dots=f_n({\bf x})=0\},
\end{equation*}
which is a subvariety of codimension 2 at least, away from which $f$ defines a holomorphic mapping $f:\p^N\setminus I(f)\rightarrow \p^N$. 
More generally, one can define the indeterminacy set $I(f)$ of a dominant rational map $f:X\rightarrow Y$ of varieties as the smallest Zariski-closed subset such that $f : X\setminus I(f) \rightarrow Y$ is a morphism.

Recall that for rational varieties $X$ and $Y$, a rational mapping $f: X \rightarrow Y$ induces by pull-back the map $f^*: \C(Y)\to \C(X)$ between the corresponding rational function fields, which then gives the maps $f^*: \Div(Y)\to \Div(X)$ of groups of divisors and $f^*: \Pic(Y)\to \Pic(X)$ of Picard groups. 
We also have the usual push-forward map $f_*:\Div(X)\to\Div(Y)$ and when $f$ is birational we have $f_*=(f^{-1})^{*}$. 
For a divisor $D$ on $Y$, it is important to distinguish between the \emph{total transform} $f^*(D)$ of $D$ under $f$ and the \emph{proper transform} (or \emph{strict transform}), which is the Zariski closure in $X$ of the set $\{f^{-1} (x)~|~x \text{ is a generic point in } D\}$.

The main idea we will use in treating rational maps is to regularise them using blow-ups, after which we can analyse the dynamics more easily via the pull-backs (respectively push-forwards) of the regularised maps, since these are linear. 
The notion of regularisation required for such an approach to be effective is provided by the concept of \emph{algebraic stability} which, roughly speaking, means that pull-back commutes with iteration. 
We have the following fundamental fact. 

\begin{proposition} \label{prop:nocontraction}
Let $f:X\to Y$ and $g:Y\to Z$ be dominant rational maps of varieties. 
Then
$f^*\circ g^* =  (g\circ f)^*$ 
holds if and only if there does not exist a prime divisor $D$ on $X$ such that $f(D\setminus I(f))\subset I(g)$.
\end{proposition}
For proofs, see \cite{dillerfavre} for the case when $X$, $Y$ and $Z$ are two-dimensional, and \cite{bayraktar,roeder,Bedford-Kim-2008} for the general case.
The following definition is due to Fornaess and Sibony \cite{FS1995}, which we make in the case relevant to this paper, i.e. when $X$ is a smooth projective variety.
\begin{definition}
    A rational map $\varphi$ from a variety $X$ to itself is called \emph{algebraically stable} or \emph{1-regular} if the pull-back $\varphi^*:\Pic(X)\rightarrow \Pic(X)$ satisfies
$(\varphi^*)^n=(\varphi^n)^*$. 
\end{definition}
Proposition \ref{prop:nocontraction} implies the following characterisation of algebraic stability.

\begin{proposition} \label{prop:nocontractionalgstable}
A dominant rational map $\varphi$ from a variety $X$ to itself is algebraically stable if and only if there does not exist a positive integer $k$ and a prime divisor $D$ on $X$ such that $f(D\setminus I(f))\subset I(f^k)$.
\end{proposition}

The results of Diller and Favre \cite{dillerfavre} imply that any birational self-map of a surface (i.e. a two-dimensional variety) $X$ can be birationally conjugated to an algebraically stable map.
That is, given $\varphi : X \rightarrow X$, one can find a surface $\widetilde{X}$ and birational map $\pi : \widetilde{X} \rightarrow X$, which is a composition of blow-ups and blow-downs, such that $\tilde{\varphi}:= \pi \circ \varphi \circ \pi^{-1}: \widetilde{X} \rightarrow \widetilde{X}$ is algebraically stable.

While a conjugation of a birational map $\varphi : X \rightarrow X$ to an algebraically stable one is not guaranteed to exist in higher dimensions (and does not in general by the results of \cite{Bell-Diller-Jonsson-Krieger-2024}, since if it did then the examples therein would not have transcendental dynamical degrees), for many integrable cases it has proven possible to obtain a conjugation under which $\varphi$ becomes not just algebraically stable, but a pseudo-automorphism as defined as follows.
\begin{definition}
    A birational mapping $f : X \rightarrow Y$ of varieties is called a pseudo-isomorphism (or pseudo-automorphism if $X=Y$) if it restricts to an isomorphism between Zariski open subsets whose complements are of codimension at least 2. That is, there exist Zariski open subsets $U\subset X$ and $V\subset Y$ such that 
    \begin{enumerate}
        \item the codimensions of $X\setminus U$ and $Y\setminus V$ are at least 2,
        \item $f|_{U} : U \rightarrow V$ is an isomorphism.
    \end{enumerate}
\end{definition}

Note that in dimension 2, the notion of pseudo-isomorphism coincides with that of isomorphism, but this is not true in higher dimensions.
Also note that a pseudo-automorphism is algebraically stable but the converse is not true in general.

In the case when, given a birational mapping, there exists a variety of which this becomes a pseudo-automorphism, the variety is sometimes referred to as a \emph{space of initial conditions}, by analogy with the algebraic surfaces on which the usual differential and discrete Painlev\'e equations are regularised in the work of Okamoto and Sakai respectively \cite{OKAMOTO1979,sakai}.
The strategy by which spaces of initial conditions for birational mappings in dimension greater than two are constructed in \cite{algebraicallystablevariety,stefan,Takenawa-2021,Takenawa-2024} (which will be adopted when we regularise the mapping \eqref{Map1} as a pseudo-automorphism in 
Section \ref{sec:regularisation}) is to blow-up along subvarieties of codimension $>1$ which are the images of codimension 1 subvarieties under $\varphi$ or $\varphi^{-1}$.



Let us briefly recall the process of blowing up along a subvariety in dimension $\geq 2$.
Consider a variety of dimension $N$, regarded as a complex manifold, with a local coordinate chart $U\subset \C^N$. 
Let $V$ be a subvariety of dimension $N-k$, $k\geq 2$, written in coordinates as 
\begin{equation*}
    V = \left\{ \left(x_1,\dots,x_N\right)\in U ~|~ x_1-h_1(x_{k+1},\dots x_N)=\dots=x_k-h_k(x_{k+1},\dots x_N)=0\right\},
\end{equation*}
for some regular functions $h_1,\dots,h_k$ on $U$. 
Blowing up along $V$ gives an open variety $X$, given as the gluing of the $k$ affine charts
\begin{align*}
U_i=\{(u_1^{(i)},\dots, u_k^{(i)},x_{k+1},\dots x_N )\in \C^N\}\qquad i=1,\dots,k,
\end{align*}
according to the birational morphism $\pi:X\to U$ defined by $\pi: U_i \to U$: 
\begin{align*}
(x_1,\dots,x_N) &= (u_1^{(i)}u_i^{(i)}+h_1, \dots, u_{i-1}^{(i)}u_i^{(i)}+h_{i-1},
u_i^{(i)}+h_i, \\
&\qquad u_{i+1}^{(i)}u_i^{(i)}+h_{i+1}, \dots, u_k^{(i)}u_i^{(i)}+h_k,x_{k+1},\dots,x_N).
\end{align*} 
It is convenient to write the coordinates of $U_i$ as
$$\left(\frac{x_1-h_1}{x_i-h_i},\dots,\frac{x_{i-1}-h_{i-1}}{x_i-h_i},x_i-h_i,
\frac{x_{i+1}-h_{i+1}}{x_i-h_i},\dots,\frac{x_k-h_k}{x_i-h_i},x_{k+1},\dots x_N \right).$$
We will sometimes refer to the subvariety $V$ as the \emph{centre of the blow-up}, and the birational morphism $\pi$ as the \emph{blow-up projection}.
The exceptional divisor $E=\pi^{-1}(V)$ of the blow-up is written as $u_i=0$ in $U_i$ and each point in the centre of the blow-up corresponds to a subvariety isomorphic to $\p^{k-1}$, points of which are parametrised as
$(x_1-h_1: \dots: x_{k-1}-h_k)$. 
Hence $E$ is locally a direct product $V \times \p^{k-1}$.
We call the copy of $\p^{k-1}$ over a point in $V$ a fibre of the exceptional divisor. 

If we begin with $(\p^1)^{\times N}$ and perform a sequence of $K$ blow-ups along $K$ subvarieties we obtain a rational variety $X$. 
The Picard group $\Pic(X)\simeq H^2(X,\Z)$ and its Poincar\'e dual $\simeq H_2(X,\Z)$ form the \emph{N\'eron-Severi bi-lattice} of $X$, i.e. the pair of free $\mathbb{Z}$-modules 
\begin{align}\label{NSbasis}
H^2(X,\Z)=\bigoplus_{i=1}^N \Z H_i \oplus \bigoplus_{k=1}^{K} \Z E_k,\quad 
H_2(X,\Z)=\bigoplus_{i=1}^N \Z h_i \oplus \bigoplus_{k=1}^{K} \Z e_k,
\end{align}
equipped with the bilinear intersection form given by
\begin{equation*}
\langle H_i, h_j\rangle = \delta_{ij},\quad 
\langle E_k, e_l\rangle = -\delta_{kl},\quad
\langle H_i, e_l\rangle = \langle E_k, h_j \rangle =0,
\end{equation*} 
where $\delta_{ij}$ is the Kronecker delta.
Here $H_i$ corresponds to the total transform under the blow-ups of a hyperplane given by the fibre of the projection from $(\p^1)^{\times N}$ onto the $i$-th factor.
Note that we abuse notation in using the same symbol $E_k$ for the exceptional divisor of the blow-up and for the corresponding element of $\Pic(X)$ or $H^2(X,\Z)$.

Let $\varphi$ be a pseudo-automorphism of $X$, and let $A$ and $B$ be matrices representing its actions by push-forward 
$$\varphi_*: H^2(X,\Z)\to  H^2(X,\Z)\quad {\rm and}\quad  \varphi_*: H_2(X,\Z)\to  H_2(X,\Z)$$
on cohomology and homology respectively in terms of the bases \eqref{NSbasis}.  
The push-forwards of a pseudo-automorphism preserve the intersection form, so for any ${\bf f}, {\bf g}\in\Z^{N+K}$ representing elements of $f\in H^2(X,\Z)$ and $g \in H_2(X,\Z)$ respectively in the bases \eqref{NSbasis}, it holds that
\begin{align*}
 \langle f,  g\rangle= {\bf f}^T J  {\bf g}, \quad
J=\begin{bmatrix}
I_N&0\\0& -I_K
\end{bmatrix},
\end{align*}
where ${*}^T$ denotes transpose and $I_m$ denotes the identity matrix of size $m$. 
Thus, $\langle \varphi_*{ f}, \varphi_*{g}\rangle= \langle {f}, { g}\rangle$ yields $A^T J B = J$ and, 
hence, 
\begin{align}\label{H2H2}
B= J (A^{-1})^T J,
\end{align} 
which is a formula for computing the push-forward action on $H_2(X,\Z)$ from that on $H^2(X,\Z)$.

In complex dimension 4, we can blow-up along subvarieties of codimension 2 (surfaces), codimension 3 (curves), and codimension 4 (points).
In an important departure from the two-dimensional case, we may blow-up along subvarieties which are distinct but whose intersection is non-empty.
Different choices of the order in which such subvarieties are blown up along will lead to varieties which are pseudo-isomorphic, but may not be isomorphic.
We will discuss this problem in Appendix~\ref{order_blow-ups} with an example.

\section{Singularities and their confinement}
One tool commonly used as a detector of integrability in discrete systems is the so-called singularity confinement criterion.
Roughly speaking, this criterion requires that a singularity appearing spontaneously while iterating a discrete system (in the sense that the inverse of the evolution ceases to be well-defined, which is interpreted to mean that memory of the initial conditions is lost) can propagate for only a \emph{finite} number of iterations, after which this the memory of initial data is recovered and the singularity is said to be confined.

Of course, one has to clarify the notion of singularity here. 
For a discrete system defined by a birational map, the singularity confinement property (meaning that all singularities are confined) is intimately related to the ability to regularise the map as a pseudo-automorphism.
Before explaining this, we make the important remark that the singularity confinement property is neither necessary nor sufficient for integrability. 
This is because, for instance, there exist linearisable birational maps with unconfined singularities and also examples with the singularity confinement property but chaotic dynamics in the sense of exponential degree growth \cite{hietarintaviallet}, i.e. non-zero algebraic entropy \cite{Bellon-Viallet}.

Consider a smooth projective variety $X$ obtained as some compactification of $\C^N$, i.e. containing $\C^N$ as an open subset, and a birational map $f : X \rightarrow X$. 
Suppose there is some hypersurface $D \subset X$ which is {\it contracted} by $f$ to a lower-dimensional subvariety. 
This contraction constitutes the kind of loss of memory of initial data referred to above, and as such it is called a \emph{singularity} of $f$. 
In the language of birational geometry, this is nothing but the hypersurface $D$ being blown down by $f$ to the subvariety $f(D)$. 

Let us introduce the set of contracted hypersurfaces
$${\mathcal E}(f)=\{D \subset  X \text{ a hypersurface}~|~ \operatorname{det}({\partial f}/{\partial x})|_{D}=0\},$$
where the vanishing of the Jacobian is indicative of contraction to a lower dimensional subvariety. 
The singularity consituted by contraction of $D\in \mathcal{E}(f)$ is said to be \emph{confined} if there exists an integer $n\geq 2$ such that $f^n(D)$ is of codimension 1.
Then the coordinates of $D$ are recovered in those of another hypersurface $f^n(D)$, and in this sense the memory of initial conditions is recovered. 
If the singularity corresponding to the contraction of $D$ is confined for every $D$ in ${\mathcal E}(f)$, we say that $f$ satisfies the singularity confinement criterion. 
Note that the existence of a confined singularity of $f : X\rightarrow X$ implies, by Proposition \ref{prop:nocontractionalgstable}, that $f$ is not algebraically stable. 
In order to achieve algebraic stability, the varieties to which hypersurfaces are contracted must be blown up and $f$ conjugated by the corresponding blow-up projection.

\subsection{Singularities of the mapping \eqref{Map1}}
Let us consider the system \eqref{Map1}, which we reproduce here for convenience: 
\begin{equation}
 \label{Map1sec3}
\begin{aligned}
    \varphi: \left(\p^1\right)^{\times 4} &\rightarrow \left(\p^1 \right)^{\times 4}, \\
    (q_1,p_1,q_2,p_2)&\mapsto(\bar{q}_1,\bar{p}_1,\bar{q}_2,\bar{p}_2),
\end{aligned}
\qquad 
\left\{\begin{array}{rcl}
\bar{q}_1&=&\frac{1+bq_2}{q_1q_2p_2},\\
\bar{p}_1&=&q_2, \\
\bar{q}_2&=&\frac{1+bq_1}{q_1p_1q_2},\\
\bar{p}_2&=&q_1. 
\end{array} \right.
\end{equation}
It has the symmetry $(q_1,p_1)\leftrightarrow (q_2,p_2)$, which greatly alleviates the computations required to verify the singularity confinement property of the mapping \eqref{Map1sec3}. 
We will work with $\left(\mathbb{P}^1\right)^{\times 4}$ with the usual charts, i.e. pairs $x$, $\frac{1}{x}$ of complex coordinates to cover each copy of $\mathbb{P}^1$, $x\in\{q_1,p_1,q_2,p_2 \}$, leading to an atlas for $\left(\mathbb{P}^1\right)^{\times 4}$ made up of $2^4=16$ charts.
The Jacobian determinant of $\varphi$ (in the coordinates $(q_1,p_1,q_2,p_2)$) is given by
$$J=\frac{\partial(\bar{q}_1,\bar{p}_1,\bar{q}_2,\bar{p}_2)}{\partial(q_1,p_1,q_2,p_2)}=\frac{(1+bq_1)(1+bq_2)}{(q_1q_2p_1p_2)^2}.$$
One can see immediately that $\{q_{1}=-1/b\}$ and $\{q_2=-1/b\}$ both give hypersurfaces contracted by $\varphi$. 
The way in which singularity confinement is verified is, for example in the case of the hypersurface $\{q_1=-1/b\}$, is to consider an initial condition $(q_1,p_1,q_2,p_2) = (-1/b+ \varepsilon,u,v,w)$ for $u,v,w\in \p^1$ generic, and compute the iterates in the limit as $\varepsilon \to 0$. 
The computations are straightforward but best performed using computer algebra. 
The return of the hypersurface $\{q_1=-1/b\}$ to codimension one under $\varphi^5$, after which its codimension will not increase again, means that this singularity is confined. 

We summarise the sequence of subvarieties through which the hypersurfaces $\{q_{1}=-1/b\}$ and $\{q_2=-1/b\}$ move under successive iterations of $\varphi$ in the \emph{singularity  patterns}
\begin{equation} \label{confiningpattern1}
\left(
\begin{array}{c}
-1/b\\
*\\
*\\
*
\end{array}\right)
\to
\left(
\begin{array}{c}
*\\
*\\
0^1\\
-1/b
\end{array}\right)\to
\left(
\begin{array}{c}
\infty^1\\
0^1\\
\infty^1\\
*
\end{array}\right)\to
\left(
\begin{array}{c}
0^1\\
\infty^1\\
*\\
\infty^1
\end{array}\right)\to
\left(
\begin{array}{c}
*\\
*\\
-1/b\\
0^1
\end{array}\right)\to
\left(
\begin{array}{c}
*\\
-1/b\\
*\\
*
\end{array}\right),
\end{equation}
and 
\begin{equation} \label{confiningpattern2}
\left(
\begin{array}{c}
*\\
*\\
-1/b\\
*
\end{array}\right)
\to
\left(
\begin{array}{c}
0^1\\
-1/b\\
*\\
*
\end{array}\right)\to
\left(
\begin{array}{c}
\infty^1\\
*\\
\infty^1\\
0^1
\end{array}\right)\to
\left(
\begin{array}{c}
*\\
\infty^1\\
0^1\\
\infty^1
\end{array}\right)\to
\left(
\begin{array}{c}
-1/b\\
0^1\\
*\\
*
\end{array}\right)\to
\left(
\begin{array}{c}
*\\
*\\
*\\
-1/b
\end{array}\right).
\end{equation}
This notation is, by now, common in the singularity analysis of birational mappings in the context of discrete integrable systems, and represents the iterates $\varphi^{k}(q_1,p_1,q_2,p_2)$.
In particular $*$ indicates a generic value, while $0^{\ell}$ and $\infty^{\ell}$ indicate that the corresponding variables are of order $O(\varepsilon^{\ell})$ and $O(\varepsilon^{-\ell})$ respectively as $\varepsilon\to0$.

Patterns corresponding to confined singularities such as in \eqref{confiningpattern1} and \eqref{confiningpattern2} are often referred to as \emph{strictly confining} patterns. 
However, these are not the only singularities we encounter. 
There are also two examples of so-called \emph{cyclic patterns}, corresponding to hypersurfaces which are periodic under $\varphi$, namely
\begin{equation} \label{cyclicpattern1}
\begin{aligned}
&\left(
\begin{array}{c}
0^1\\
*\\
*\\
*
\end{array}\right)
\to
\left(
\begin{array}{c}
\infty^1\\
*\\
\infty^1\\
0^1
\end{array}\right)\to
\left(
\begin{array}{c}
*\\
\infty^1\\
0^1\\
\infty^1
\end{array}\right)\to
\left(
\begin{array}{c}
*\\
0^1\\
*\\
*
\end{array}\right)\to
\left(
\begin{array}{c}
*\\
*\\
\infty^1\\
*
\end{array}\right)\to
\left(
\begin{array}{c}
*\\
\infty^1\\
0^1\\
*
\end{array}\right)\to
\\
&\left(
\begin{array}{c}
\infty^1\\
0^1\\
*\\
*
\end{array}\right)
\to
\left(
\begin{array}{c}
0^1\\
*\\
\infty^1\\
\infty^1
\end{array}\right)\to
\left(
\begin{array}{c}
*\\
\infty^1\\
*\\
0^1
\end{array}\right)\to
\left(
\begin{array}{c}
\infty^1\\
*\\
0^1\\
*
\end{array}\right)\to
\left(
\begin{array}{c}
*\\
0^1\\
\infty^1\\
\infty^1
\end{array}\right)\to
\left(
\begin{array}{c}
0^1\\
\infty^1\\
*\\
*
\end{array}\right)\to
\\
&\left(
\begin{array}{c}
\infty^1\\
*\\
*\\
0^1
\end{array}\right)\to
\left(
\begin{array}{c}
*\\
*\\
*\\
\infty^1
\end{array}\right)\to
\left(
\begin{array}{c}
0^1\\
*\\
*\\
*
\end{array}\right),
\end{aligned}
\end{equation}
and 
\begin{equation} \label{cyclicpattern2}
\begin{aligned}
&\left(
\begin{array}{c}
*\\
*\\
0^1\\
*
\end{array}\right)
\to
\left(
\begin{array}{c}
\infty^1\\
0^1\\
\infty^1\\
*
\end{array}\right)\to
\left(
\begin{array}{c}
0^1\\
\infty^1\\
*\\
\infty^1
\end{array}\right)\to
\left(
\begin{array}{c}
*\\
*\\
*\\
0^1
\end{array}\right)\to
\left(
\begin{array}{c}
\infty^1\\
*\\
*\\
*
\end{array}\right)\to
\left(
\begin{array}{c}
0^1\\
*\\
*\\
\infty^1
\end{array}\right)\to
\\
&\left(
\begin{array}{c}
*\\
*\\
\infty^1\\
0^1
\end{array}\right)
\to
\left(
\begin{array}{c}
\infty^1\\
\infty^1\\
0^1\\
*
\end{array}\right)\to
\left(
\begin{array}{c}
*\\
0^1\\
*\\
\infty^1
\end{array}\right)\to
\left(
\begin{array}{c}
0^1\\
*\\
\infty^1\\
*
\end{array}\right)\to
\left(
\begin{array}{c}
\infty^1\\
\infty^1\\
*\\
0^1
\end{array}\right)\to
\left(
\begin{array}{c}
*\\
*\\
0^1\\
\infty^1
\end{array}\right)\to
\\
&\left(
\begin{array}{c}
*\\
0^1\\
\infty^1\\
*
\end{array}\right)\to
\left(
\begin{array}{c}
*\\
\infty^1\\
*\\
*
\end{array}\right)\to
\left(
\begin{array}{c}
*\\
*\\
0^1\\
*
\end{array}\right).
\end{aligned}
\end{equation}

\section{Regularisation of mapping \eqref{Map1} as a pseudo-automorphism} \label{sec:regularisation}

The results of computations of singularity confinement behaviour above serve as a guide as to what blow-ups to perform in order to regularise the mapping \eqref{Map1} as a pseudo-automorphism.
We shall prove by computation the following.

\begin{proposition} \label{prop:pseudoautomorphism}
We have a rational variety $X$ and birational morphism $\pi : X \rightarrow (\p^1)^{\times 4}$ such that $\tilde{\varphi} = \pi^{-1} \circ \varphi \circ \pi$ is a pseudo-automorphism of $X$:
\begin{equation*}
\begin{tikzcd}[ampersand replacement=\&]
 X \arrow[r, "\tilde{\varphi}"]  \arrow[d, "\pi"] 		\& X  \arrow[d, "\pi"]  \\
(\p^1)^{\times 4} \arrow[r, "\varphi"] 					\&(\p^1)^{\times 4} .
\end{tikzcd}
\end{equation*}
The rational variety $X$ is obtained from $(\mathbb{P}^1)^{\times 4}$ through a sequence of 28 blow-ups along subvarieties of codimensions either two or three, with $\pi$ being the composition of the corresponding blow-up projections.
\end{proposition}

Before giving the subvarieties to blow up along to obtain $X$, let us illustrate some of the extra considerations that arise in this kind of calculation in dimension greater than two. 
These come from the fact that the subvarieties to be blown up may have inclusion relations among them, which does not occur in the two-dimensional case in which only points are blown up.

Inspection of the cyclic pattern in \eqref{cyclicpattern2} indicates the need to blow up, in particular, the subvarieties $\{q_1=0,\,q_2=\infty\}$ and $\{q_1=0,\, p_2=\infty\}$, both of codimension two, since these are in the orbit of the hypersurface $\{q_2=0\}$ under $\varphi$.
On the other hand, from the other cyclic pattern \eqref{cyclicpattern1} we see that the codimension three subvariety $\{q_1=0,\,q_2=\infty,\,p_2=\infty\}$ is in the orbit of $\{q_1=0\}$.
This is the intersection of the pair of codimension two subvarieties noted above from \eqref{cyclicpattern2}, and we must choose in which order to perform the three blow-ups along these. 
We make some remarks about the effect of changing orders of blow-ups in Appendix \ref{order_blow-ups}, but for now we choose, in case when we need to blow up two surfaces as well as the curve given by their intersection, to blow up first along the curve and then along the proper transforms of the surfaces under this.

For example, the subvarieties described above are denoted (according to the numbering given below in section \ref{subsec:blow-ups}) by
\begin{equation*}
V_1 : \left(\begin{array}{c} q_1 \\ p_1 \\ q_2 \\ p_2  \end{array}\right)=\left(\begin{array}{c} 0 \\ * \\ \infty \\ \infty  \end{array}\right),
\qquad 
V_4 : \left(\begin{array}{c} q_1 \\ p_1 \\ q_2 \\ p_2  \end{array}\right)=\left(\begin{array}{c} 0 \\ * \\ \infty \\ *  \end{array}\right), 
\qquad V_5 : \left(\begin{array}{c} q_1 \\ p_1 \\ q_2 \\ p_2  \end{array}\right)=\left(\begin{array}{c} 0 \\ * \\ * \\ \infty  \end{array}\right),
\end{equation*}
so $V_1 = V_4 \cap V_5$.
We blow up along the line $V_1$ first, denoting the exceptional divisor by $E_1$.
Then we consider $V_4$ and $V_5$ as the proper transforms under the blow-up along $V_1$, and proceed to blow up along these.
We illustrate this in Figure \ref{fig:blow-upalongintersection}.

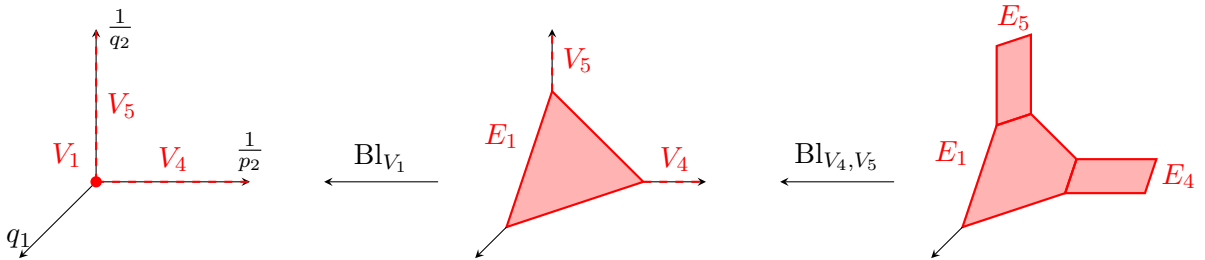
\begin{figure}[htb]
    \centering
        \begin{tikzpicture}[scale=.3,x=0.5cm,y=0.5cm,z=0.25cm,>=stealth]
\draw[->] (xyz cs:x=0) -- (xyz cs:x=13.5) node[above] {$\frac{1}{p_2}$};
\draw[->] (xyz cs:y=0) -- (xyz cs:y=13.5) node[right] {$\frac{1}{q_2}$};
\draw[->] (xyz cs:z=+0) -- (xyz cs:z=-13.5) node[above] {$q_1$};
\draw[red, thick, dashed] (xyz cs:x=0) -- (xyz cs:x=13.5) node[midway, above] {$V_4$};
\draw[red, thick, dashed] (xyz cs:y=0) -- (xyz cs:y=13.5) node[midway, right] {$V_5$};

\node[fill,red,circle,inner sep=1.5pt,label={above left:{\color{red} $V_1$}}] at (xyz cs:z=+0) {};

\draw[->] (xyz cs:x=30) -- (xyz cs:x=20) node[midway, above] {$\operatorname{Bl}_{V_1}$};

\begin{scope}[xshift=20cm]
\draw[->] (xyz cs:x=30) -- (xyz cs:x=20) node[midway, above] {$\operatorname{Bl}_{V_4, V_5}$};

\draw[thick,red] (xyz cs:x=8) -- (xyz cs:y=8)-- node[midway, above left] {$E_1$} (xyz cs:z=-8)--(xyz cs:x=8)  ;
\filldraw[red,opacity=.3] (xyz cs:x=8) -- (xyz cs:y=8)-- (xyz cs:z=-8)--(xyz cs:x=8)  ;
\draw[->] (xyz cs:x=8) -- (xyz cs:x=13.5) 
;
\draw[->] (xyz cs:y=8) -- (xyz cs:y=13.5) 
;
\draw[->] (xyz cs:z=-8) -- (xyz cs:z=-13.5) 
;

\draw[red, thick, dashed] (xyz cs:x=8) -- (xyz cs:x=13.5) node[midway, above] {$V_4$};
\draw[red, thick, dashed] (xyz cs:y=8) -- (xyz cs:y=13.5) node[midway, right] {$V_5$};

\end{scope}

\begin{scope}[xshift=40cm]

\draw[thick, red] (xyz cs:x=6,y=2,z=0) -- (xyz cs:x=2,y=6,z=0) --  (xyz cs:x=0,y=6,z=-2) -- (xyz cs:x=0,y=0,z=-8) node[midway, above left] {$E_1$} -- (xyz cs:x=0,y=0,z=-8) -- (xyz cs:x=6,y=0,z=-2) --(xyz cs:x=6,y=2,z=0)  ;
\filldraw[thick,red,opacity=0.3] (xyz cs:x=6,y=2,z=0) -- (xyz cs:x=2,y=6,z=0) --  (xyz cs:x=0,y=6,z=-2) -- (xyz cs:x=0,y=0,z=-8)-- (xyz cs:x=0,y=0,z=-8) -- (xyz cs:x=6,y=0,z=-2) --(xyz cs:x=6,y=2,z=0)  ;

\draw[thick, red] (xyz cs:x=6,y=2,z=0) -- (xyz cs:x=13,y=2,z=0)  --  (xyz cs:x=13,y=0,z=-2) node[midway, right] {$E_4$} -- (xyz cs:x=6,y=0,z=-2) -- (xyz cs:x=6,y=2,z=0);
\filldraw[thick, red, opacity=0.3] (xyz cs:x=6,y=2,z=0) -- (xyz cs:x=13,y=2,z=0)  -- (xyz cs:x=13,y=0,z=-2) -- (xyz cs:x=6,y=0,z=-2) -- (xyz cs:x=6,y=2,z=0);
\draw[thick, red] (xyz cs:x=2,y=6,z=0)  -- (xyz cs:x=2,y=13,z=0) --   (xyz cs:x=0,y=13,z=-2)node[midway, above] {$E_5$}  -- (xyz cs:x=0,y=6,z=-2) -- (xyz cs:x=2,y=6,z=0) ;
\filldraw[thick, red, opacity=0.3] (xyz cs:x=2,y=6,z=0)  -- (xyz cs:x=2,y=13,z=0) -- (xyz cs:x=0,y=13,z=-2) -- (xyz cs:x=0,y=6,z=-2) -- (xyz cs:x=2,y=6,z=0) ;

;
\draw[->] (xyz cs:z=-8) -- (xyz cs:z=-13.5) 
;

\end{scope}

\end{tikzpicture}
    \caption{Order of blow-ups along two codimension-2 subvarieties and along their intersection.}
    \label{fig:blow-upalongintersection}
\end{figure}

\subsection{The sequence of blow-ups} \label{subsec:blow-ups}

The birational morphism $\pi : X\rightarrow (\p^1)^{\times 4}$ in Proposition \ref{prop:pseudoautomorphism} is a composition of a sequence of 28 blow-ups. 
These can be grouped into 4 groups of 7, which are related by the transformations $(q_1,p_1,q_2,p_2)\rightarrow(q_2,p_2,q_1,p_1)$ and $(q_1,p_1,q_2,p_2)\rightarrow(p_1,q_1,p_2,q_2)$. 
The first group is given in Figure \ref{fig:firstsevenblow-ups} below, where we have indicated inclusion relations among the subvarieties and an arrow $V_2 \leftarrow V_3$ indicates that the subvariety $V_3$ lies inside the exceptional divisor of the blow-up of $V_2$.

\begin{figure}[h]
    \centering
    \begin{tikzpicture}    
 \node at (0,0)
 {
 \begin{tikzcd}
    [column sep=-0.25in]
        &       &V_1 : \left(\begin{array}{c} 0 \\ * \\ \infty \\ \infty  \end{array}\right)   &       & V_2 : \left(\begin{array}{c} 0 \\ \infty \\ * \\ \infty  \end{array}\right)  &  & \quad \arrow[ll, phantom, "\leftarrow"] & \\
        &V_4 : \left(\begin{array}{c} 0 \\ * \\ \infty \\ *  \end{array}\right) \arrow[ur, phantom,sloped, "\supset"]  &       &V_5 : \left(\begin{array}{c} 0 \\ * \\ * \\ \infty  \end{array}\right) \arrow[ul, phantom,sloped, "\subset"] \arrow[ur, phantom,sloped, "\supset"] &       &V_6 : \left(\begin{array}{c} 0 \\ \infty \\ * \\ *  \end{array}\right)\arrow[ul, phantom,sloped, "\subset"] &  \qquad \qquad & V_7: \left(\begin{array}{c} 0 \\ -b^{-1} \\ * \\ * \end{array}\right)
    \end{tikzcd}
    };
    \node at (5,1.4) 
    {
    $V_3 : \left(\begin{array}{c} q_2 + \frac{b}{q_1 p_1}  \\ \frac{1}{q_1 p_2} \\ q_1 \\ q_2  \end{array}\right) = \left(\begin{array}{c} 0 \\ * \\ 0 \\ * \end{array}\right)$
    };
\end{tikzpicture}
    \caption{The first seven blow-ups in the construction of $X$.}
    \label{fig:firstsevenblow-ups}
\end{figure}
The way in which $V_3$ is given in Figure \ref{fig:firstsevenblow-ups} corresponds to local equations in charts introduced according to the convention in Section \ref{sec:blow-ups,pseudo-automorphisms} is as follows.
In the local chart with coordinates $(q_1, P_1,q_2,P_2) = (q_1, \tfrac{1}{p_1}, q_2, \tfrac{1}{p_2})$ on $(\p^1)^{\times 4}$, $V_2$ is given in terms of local equations by
\begin{equation*}
    V_2 :\quad  q_1 = 0, P_1 = 0, P_2 = 0, 
\end{equation*}
so three charts introduced to cover the part of the exceptional divisor away from $\{q_2=\infty\}$ according to the convention in Section \ref{sec:blow-ups,pseudo-automorphisms} can be written with some changes in notation for convenience as $(r_1,s_1,t_1,q_2)$, $(r_2,s_2,t_2,q_2)$ and $(r_3,s_3,t_3,q_2)$, where under the blow-up projection we have
\begin{equation*}
    q_1 = r_1 t_1 = r_2 t_2 = t_3, \quad P_1 = s_1 t_1 = t_2 = r_3 t_3, \quad P_2 = t_1 = s_2 t_2 = s_3 t_3,
\end{equation*}
so the exceptional divisor has local equation $t_i=0$.
Alternatively, we have the formal expressions
\begin{equation*}
\begin{aligned}
    \left( r_1,s_1,t_1,q_2\right) &= \left( \frac{q_1}{P_2}, \frac{p_1}{P_2}, P_2, q_2 \right) =  \left( q_1p_2 , p_1 p_2, \frac{1}{p_2}, q_2 \right), \\
    \left( r_2,s_2,t_2,q_2\right) &= \left( \frac{q_1}{P_1}, \frac{P_2}{P_1}, P_1, q_2 \right) =  \left( q_1p_1 , \frac{p_1}{p_2}, \frac{1}{p_1}, q_2 \right), \\
    \left( r_3,s_3,t_3,q_2\right) &= \left( \frac{P_1}{q_1}, \frac{P_2}{q_1}, q_1, q_2 \right) =  \left( \frac{1}{q_1p_1} , \frac{1}{q_1 p_2}, q_1, q_2 \right). 
\end{aligned}
\end{equation*}
Then the specification of $V_3$ as in Figure \ref{fig:firstsevenblow-ups} should be interpreted that it has local equations in, for example, the chart $(r_3,s_3,t_3,q_2)$ given by $q_2+b r_3 = t_3 = 0$.

We specify the remaining three groups of seven blow-ups on Figures \ref{fig:secondsevenblow-ups}, \ref{fig:thirdsevenblow-ups} and \ref{fig:fourthsevenblow-ups}.
\begin{figure}[h]
    \centering
    \begin{tikzpicture}    
 \node at (0,0)
 {
 \begin{tikzcd}
    [column sep=-0.25in]
        &       &V_8 : \left(\begin{array}{c} * \\ 0 \\ \infty \\ \infty  \end{array}\right)   &       & V_{9} : \left(\begin{array}{c} \infty \\ 0 \\ \infty \\ *  \end{array}\right)   &  & \quad \arrow[ll, phantom, "\leftarrow"] & \\
        &V_{11} : \left(\begin{array}{c} * \\ 0 \\ * \\ \infty  \end{array}\right) \arrow[ur, phantom,sloped, "\supset"]  &       &V_{12} : \left(\begin{array}{c} * \\ 0 \\ \infty \\ *  \end{array}\right) \arrow[ul, phantom,sloped, "\subset"] \arrow[ur, phantom,sloped, "\supset"] &       &V_{13} : \left(\begin{array}{c} \infty \\ 0 \\ * \\ *  \end{array}\right)\arrow[ul, phantom,sloped, "\subset"] &  \qquad \qquad & V_{14} : \left(\begin{array}{c} - b^{-1} \\ 0 \\ * \\ * \end{array}\right)
    \end{tikzcd}
    };
    \node at (5,1.4) 
    {
    $V_{10} : \left(\begin{array}{c} p_2 + \frac{b}{q_1 p_1}  \\ \frac{1}{p_1 q_2} \\ p_1 \\ p_2  \end{array}\right) = \left(\begin{array}{c} 0 \\ * \\ 0 \\ * \end{array}\right)$
    };
\end{tikzpicture}
    \caption{The second seven blow-ups in the construction of $X$.}
    \label{fig:secondsevenblow-ups}
\end{figure}




\begin{figure}[h]
    \centering
    \begin{tikzpicture}    
 \node at (0,0)
 {
 \begin{tikzcd}
    [column sep=-0.25in]
        &       &V_{15} : \left(\begin{array}{c} \infty \\ \infty \\ 0 \\ *  \end{array}\right)   &       & V_{16} : \left(\begin{array}{c} * \\ \infty \\ 0 \\ \infty  \end{array}\right)   &  & \quad \arrow[ll, phantom, "\leftarrow"] & \\
        &V_{18} : \left(\begin{array}{c} \infty \\ * \\ 0 \\ *  \end{array}\right) \arrow[ur, phantom,sloped, "\supset"]  &       &V_{19} : \left(\begin{array}{c} * \\ \infty \\ 0 \\ *  \end{array}\right) \arrow[ul, phantom,sloped, "\subset"] \arrow[ur, phantom,sloped, "\supset"] &       &V_{20} : \left(\begin{array}{c} * \\ * \\ 0 \\ \infty  \end{array}\right)\arrow[ul, phantom,sloped, "\subset"] &  \qquad \qquad & V_{21} : \left(\begin{array}{c}  * \\ * \\ 0 \\ -b^{-1}\end{array}\right)
    \end{tikzcd}
    };
    \node at (5,1.4) 
    {
    $V_{17} : \left(\begin{array}{c} q_1 + \frac{b}{q_2 p_2}  \\ \frac{1}{p_1 q_2} \\ q_2 \\ q_1  \end{array}\right) = \left(\begin{array}{c} 0 \\ * \\ 0 \\ * \end{array}\right)$
    };
\end{tikzpicture}
    \caption{The third seven blow-ups in the construction of $X$.}
    \label{fig:thirdsevenblow-ups}
\end{figure}




\begin{figure}[h]
    \centering
    \begin{tikzpicture}    
 \node at (0,0)
 {
 \begin{tikzcd}
    [column sep=-0.25in]
        &       &V_{22} : \left(\begin{array}{c} \infty \\ \infty \\ * \\ 0  \end{array}\right)   &       & V_{23} : \left(\begin{array}{c} \infty \\ * \\ \infty \\ 0  \end{array}\right)   &  & \quad \arrow[ll, phantom, "\leftarrow"] & \\
        &V_{25} : \left(\begin{array}{c} * \\ \infty \\ * \\ 0  \end{array}\right) \arrow[ur, phantom,sloped, "\supset"]  &       &V_{26} : \left(\begin{array}{c} \infty \\ * \\ * \\ 0  \end{array}\right) \arrow[ul, phantom,sloped, "\subset"] \arrow[ur, phantom,sloped, "\supset"] &       &V_{27} : \left(\begin{array}{c} * \\ * \\ \infty \\ 0 \end{array}\right)\arrow[ul, phantom,sloped, "\subset"] &  \qquad \qquad & V_{28} : \left(\begin{array}{c}  * \\ * \\  -b^{-1} \\ 0\end{array}\right)
    \end{tikzcd}
    };
    \node at (5,1.4) 
    {
    $V_{24} : \left(\begin{array}{c} p_1 + \frac{b}{q_2 p_2}  \\ \frac{1}{q_1 p_2} \\ p_2 \\ p_1  \end{array}\right) = \left(\begin{array}{c} 0 \\ * \\ 0 \\ * \end{array}\right)$
    };
\end{tikzpicture}
    \caption{The fourth seven blow-ups in the construction of $X$.}
    \label{fig:fourthsevenblow-ups}
\end{figure}


Verifying that the map $\varphi$ becomes a pseudo-automorphism of $X$ when lifted under the blow-ups is a long but standard computation in local coordinates. 
For concrete examples of this kind of computation we refer the reader to, e.g., \cite{stefan}. 

\subsection{Linear action on $H^2(X,\Z)$ and degree growth}

As noted in Section \ref{sec:blow-ups,pseudo-automorphisms}, regularisation of $\varphi$ as a pseudo-automorphism allows us to analyse its dynamics via the induced linear action on the Picard group.
 The variety $X$ obtained by the sequence of blow-ups outlined above has 
    \begin{equation} \label{PicXbasis}
    H^2(X,\mathbb{Z}) \simeq \Pic(X) \simeq \operatorname{Cl}(X)=\Z H_{q_1} \oplus \Z H_{p_1} \oplus \Z H_{q_2} \oplus \Z H_{p_2} \oplus \bigoplus_{i=1}^{28} \Z E_i,
    \end{equation}
where $E_i$ is the linear equivalence class of the exceptional divisor of the blow-up along $V_i$ (or more precisely its total transform under any subsequent blow-ups), and $H_{q_i}, H_{p_i}$ are classes of total transforms of divisors $\{q_i=\text{const} \}$, $\{p_i=\text{const} \}$ on $(\p^1)^{\times 4}$, $(i=1,2)$. 
The computations which establish Proposition \ref{prop:pseudoautomorphism} also lead to Lemma \ref{lem:pushforward} below.
\begin{notation}
    From this point onward, for the sake of brevity, we will denote sums of classes of exceptional divisors by $$E_{i_1,i_2,i_3,...,i_n}=E_{i_1}+E_{i_2}+...+E_{i_n}.$$
\end{notation}
\begin{lemma} \label{lem:pushforward}
The pushforward action $\tilde{\varphi}_* : \Pic(X)\rightarrow \Pic(X)$ of the mapping in Proposition \ref{prop:pseudoautomorphism} is given by 
\begin{equation} \label{pushforwardcohom}
    \tilde{\varphi}_* : \left\{ 
    \begin{aligned}
    &\begin{aligned}
        H_{q_1} &\mapsto H_{p_2},\quad H_{p_1}\mapsto H_{q_2}+H_{p_1}+H_{p_2}-E_{8,9,12,15,16,19,21,22,23,25,27}, \\
        H_{q_2} &\mapsto H_{p_1}, \quad H_{p_2}\mapsto H_{q_1}+H_{p_1}+H_{p_2}-E_{1,2,5,7,8,9,11,13,22,23,26},
    \end{aligned}
    \\
    &\begin{aligned}
        E_1&\mapsto E_{25}, &~ &E_2\mapsto H_{p_2}-E_{22,23,25,26,27},  &~ &E_3\mapsto E_{28},\\
        E_4&\mapsto E_{22}, &~ &E_5\mapsto E_{27}, \qquad E_6\mapsto E_{26}, &~ &E_7\mapsto E_{24},\\
        E_8&\mapsto E_6,    &~ &E_9\mapsto E_2, \qquad E_{10}\mapsto E_3,  &~ &E_{11}\mapsto E_4, 
    \end{aligned}
    \\
    &\begin{aligned}
        E_{12}&\mapsto H_{p_1}-E_{2,6,15,16,19,22,25},  &~ &E_{13}\mapsto E_1, &~ &E_{14}\mapsto H_{p_2}-E_{21},
    \end{aligned}
     \\
    &\begin{aligned}
        E_{15}&\mapsto E_{11},  &~ &E_{16}\mapsto H_{p_1}-E_{8,9,11,12,13},             &~ &E_{17}\mapsto E_{14},\\
        E_{18}&\mapsto E_8,     &~ &E_{19}\mapsto E_{13}, \qquad E_{20}\mapsto E_{12},  &~ &E_{21}\mapsto E_{10},\\
        E_{22}&\mapsto E_{20},  &~ &E_{23}\mapsto E_{16}, \qquad E_{24}\mapsto E_{17},  &~ &E_{25}\mapsto E_{18},\\
    \end{aligned}
     \\
    &\begin{aligned}    
        E_{26}&\mapsto H_{p_2}-E_{1,2,5,8,11,16,20},    &~ &E_{27}\mapsto E_{15}, &~ &E_{28}\mapsto H_{p_1}- E_7,
    \end{aligned}
    \end{aligned}
    \right.
\end{equation}

\end{lemma}

The first insight into the integrability of the mapping \eqref{Map1} that comes from our construction of $X$ is the computation of the rate of growth of the degrees of the iterates $\varphi^n$. 
For a birational map $f : \C^N \rightarrow \C^N$ written as $f (x_1,\dots,x_N) = \left(\frac{P_1(x_1,\dots, x_N)}{Q_1(x_1,\dots, x_N)}, \dots, \frac{P_N(x_1,\dots, x_N)}{Q_N(x_1,\dots, x_N)} \right)$, where for each $i=1,\dots,N$, $P_i$, $Q_i$ are polynomials in $x_1,\dots,x_N$ with no common factors, the degree of $f$ is 
\begin{equation*}
    \operatorname{deg} f = \underset{i\in\{1,\dots,N\}}{\operatorname{max}} \left\{ \operatorname{max} \left\{ \operatorname{deg} P_i,\operatorname{deg} Q_i \right\} \right\}, 
\end{equation*}
where $\operatorname{deg} P_i$ and $\operatorname{deg} Q_i$ are the total degrees as polynomials in $\C[x_1,\dots,x_N]$ of $P_i$ and $Q_i$ respectively.
The rate of growth of the sequence of degrees
   $d_n = \operatorname{deg} f^n$
of iterates is a measure of integrability of the discrete dynamical system defined by $f$, and the quantity 
\begin{equation*}
    \epsilon = \lim_{n\to\infty} \frac{1}{n}\log d_n,
\end{equation*}
is called the \emph{algebraic entropy} \cite{Bellon-Viallet} of $f$.
We remark that it is more standard to define $\epsilon$ in terms of the birational self-map of $\p^N$ obtained by extending $f$ from $\C^N$ to $\p^N$ via homogeneous coordinates $[X_0:X_1:\dots:X_N]=[1:x_1:\dots:x_N]$, and in this case the exact sequence of degrees will be different but its rate of growth, as well as $\epsilon$, will be the same.
One says that a discrete system defined by a birational map is integrable if the rate of degree growth is sub-exponential, i.e. $\epsilon=0$.
The sequence of degrees $d_n$ as defined above and the growth rate of $\varphi$ can be computed using the action of $\tilde{\varphi}_*$ on $\Pic(X)\simeq H^2(X,\Z)$.

\begin{proposition}
    The degree growth of the map \eqref{Map1} is quadratic, i.e. $d_n\sim n^2$ as $n\to +\infty$.
\end{proposition}
\begin{proof}
    Computing the Jordan normal form of the matrix representing $\varphi_*$ with respect to the basis $(H_{q_1},H_{p_1},H_{q_2},H_{p_2}, E_1,\dots, E_{28})$ using \eqref{pushforwardcohom}, one sees that all eigenvalues have modulus one, and the only block of size greater than one is 
    $$ \left(\begin{array}{ccc}1 & 1 & 0  \\0 & 1 & 1  \\0 & 0 & 1 \end{array}\right).$$
    Therefore the coefficients from the expressions for $\tilde{\varphi}_*^n(H_{q_1})$, $\tilde{\varphi}_*^n(H_{p_1})$, $\tilde{\varphi}_*^n(H_{q_2})$, $\tilde{\varphi}_*^n(H_{p_2})$ in terms of the basis $(H_{q_1},H_{p_1},H_{q_2},H_{p_2}, E_1,\dots, E_{28})$ will be polynomial functions of $n$ with coefficients being linear combinations of $n$-th powers of roots of unity, and the degrees of these polynomials will be at most two.
    Therefore the degree growth is quadratic.
\end{proof}

\subsection{Anticanonical linear system of $X$, conserved quantities, and Hamiltonian flows}

The anticanonical divisor class of $X$ is given in terms of the generators of $\Pic(X)$ in \eqref{PicXbasis} as
\begin{equation*}
-{K}_X=2H_{q_1}+2H_{q_2}+2H_{p_1}+2H_{p_2}-2E_{1,2,8,9,15,16,22,23}-\sum_{i\neq 1,2,8,9,15,16,22,23 }E_i.
\end{equation*}
By direct calculation we have the following.
\begin{proposition}
    The anticanonical linear system of $X$ is two-dimensional and can be written as
$$|- K_X | = \left\{  L_{\lambda} ~|~ \lambda= [\lambda_0:\lambda_1:\lambda_2] \in \p^2 \right\}, $$
whose generic element is the irreducible curve defined by
\begin{equation*}
\begin{gathered}
L_{\lambda} = \text{ proper transform of } \left\{\lambda_0 I_0 + \lambda_1 I_1 + \lambda_2 I_2 = 0\right\} , \\
\begin{aligned}
I_0 &= q_1 p_1 q_2 p_2, \\
I_1 &= q_1 p_1 + q_2 p_2 + b \left( q_1 p_1 p_2 + p_1 q_2 p_2 + q_1 p_1 q_2 + q_1 q_2 p_2 \right) + q_1^2 p_1 q_2 p_2^2 + q_1 p_1^2 q_2^2 p_2, \\
I_2 &= 1 + q_1^2 p_1^2 q_2^2 p_2^2 + b \left( 1 + q_1 p_1 q_2 p_2 \right) \left( q_1 + p_1 + q_2 + p_2 \right) + b^2 \left( q_1 q_2 + q_1 p_2 + p_1 q_2 + p_1 p_2\right). 
\end{aligned}
\end{gathered}
\end{equation*}
\end{proposition}
Note the parallel with the case of the two-dimensional Quispel-Roberts-Thompson (QRT) maps, which become automorphisms of rational elliptic surfaces, with anticanonical linear systems providing their elliptic fibrations.

The pseudo-automorphism $\tilde{\varphi} : X \rightarrow X$ preserves each element of $|- K_X |$, and the action of the mapping $\varphi$ on the polynomials $I_j$ providing the basis of the linear system is
$$\varphi(I_j)=\frac{(1+bq_1)(1+bq_2)}{(q_1a_2p_1p_2)^2}I_j\qquad \text{ for } \quad j=0,1,2.$$
In particular this leads to the following, which is verified by direct calculation.
\begin{corollary}
    Ratios of $I_0,I_1,I_2$, e.g. $H_1 = \frac{I_1}{I_0}$, $H_2 = \frac{I_2}{I_0}$, are conserved quantities of $\varphi$. 
\end{corollary}

Among the effective anticanonical divisors $D\in |-K_X|$ is that corresponding to the poles of the 4-form $\frac{dq_1\wedge dp_1 \wedge dq_2 \wedge dp_2}{q_1 p_1 q_2 p_2}$,
which decomposes as 
\begin{equation}\label{decomp}
\begin{aligned}
 D =\sum_{i=1}^{28}D_i &\equiv (H_{q_1}-E_{1,2,4,5,6,7}) + (H_{q_1}-E_{9,13,15,18,22,23,26})
\\
&\qquad + (H_{q_2}-E_{8,9,11,12,13,14}) +(H_{q_2}-E_{1,4,8,9,12,23,27})\\
&\qquad + (H_{p_1}-E_{15,16,18,19,20,21}) + (H_{p_1}-E_{2,6,15,16,19,22,25})
\\
&\qquad+(H_{p_2}-E_{22,23,25,26,27,28}) + (H_{p_2}-E_{1,2,5,8,11,16,20})\\
&\qquad + (E_1) + (E_2-E_3) + (E_4) + (E_5) + (E_6)+
\\
&\qquad+ (E_8) + (E_9-E_{10}) + (E_{11}) + (E_{12}) + (E_{13})\\
&\qquad+ (E_{15}) + (E_{16}-E_{17}) + (E_{18}) + (E_{19}) + (E_{20})
\\
&\qquad+ (E_{22}) + (E_{23}-E_{24}) + (E_{25}) + (E_{26}) + (E_{27}).
\end{aligned}
\end{equation}
This also provides the divisor of poles of the 2-form on $X$ defined in coordinates by
\begin{equation} \label{symplecticform}
\omega=\frac{dq_1\wedge dp_1}{q_1p_1}+ \frac{dq_2 \wedge dp_2}{q_2p_2},
\end{equation}
which is preserved by $\tilde{\varphi}$. 
This is holomorphic away from $D$ and provides a symplectic form written in log-canonical coordinates (equivalent to that $\log(q_i ), \log(p_i ), i=1, 2$, are canonical coordinates), so that $\omega\wedge \omega$ is the 4-form giving $D$ in \eqref{decomp}.
\begin{remark}
From the conserved quantities one can construct two integrable {\it continuous commuting} Hamiltonian flows.
With the symplectic form \eqref{symplecticform}, the canonical tangent vector for a function $f$ is
$$X_f=\sum_{i=1,2}q_ip_i(\partial_{p_i}f\partial_{q_i}-\partial_{q_i}f\partial_{p_i})\quad X_f\omega=df$$
and the Poisson bracket is given by
$$\{f,g\}=\sum_{i=1,2}q_ip_i(\partial_{p_i}f\partial_{q_i}g-\partial_{q_i}f\partial_{p_i}g).$$
The two conserved quantities give the Hamiltonian flows
$$\dot{q_i}=q_ip_i\partial_{p_i}H_j,\quad \dot{p_i}=-q_ip_i\partial_{q_i}H_j,\quad i,j=1,2$$
and are in involution with respect to this Poisson bracket, i.e. $\{H_1 , H_2 \} = 0$.
\end{remark}

\section{Deautonomisation and 4D discrete Painlev\'e equations}

In the two-dimensional case, there are deautonomisation procedures which produce discrete Painlev\'e equations from integrable autonomous maps. 
A geometric formalism of such a procedure was given in \cite{fane1}, and in this section we adapt this to the four-dimensional case at hand similarly to as done in \cite{stefan}.

In dimension two, one considers an integrable autonomous discrete system defined by a birational map that becomes an automorphism of a rational elliptic surface, say $S$, with its elliptic fibration provided by a pencil of effective anticanonical divisors in $|-K_S|$.
For each choice of fibre (either elliptic or singular), one can perturb the locations of points blown up to obtain $S$ such that only the chosen fibre survives as an effective anticanonical divisor.
That is, embed $S$ into a parametrised family $S_{\boldsymbol{a}}$ such that surfaces in the family,  apart from $S$, have $\operatorname{dim}|-K_{S_{\boldsymbol{a}}}|=0$, with a single effective anticanonical divisor of the same type as the chosen fibre.
Here $\boldsymbol{a}$ is a tuple of parameters controlling locations of points to be blown up to obtain $S_{\boldsymbol{a}}$. 
Then from the family $S_{\boldsymbol{a}}$ one constructs discrete Painlev\'e equations from symmetries as in the Sakai framework \cite{sakai}, which act as automorphisms of the family of surfaces.
These are described on the level of the Picard group of the surfaces as Cremona isometries.
In particular, a discrete Painlev\'e equation corresponding to the same Cremona isometry as the original autonomous map is regarded as a deautonomised version of it, with the autonomous map recovered by specialising parameters in the way that $S$ is recovered from $S_{\boldsymbol{a}}$.

In four dimensions we do not have the theory of rational elliptic surfaces and the Kodaira classification of singular fibres to work with, but we may still construct, from the variety $X$ and a choice of $D\in |-K_X|$, a parametrised family of varieties $X_{\boldsymbol{b}}$ such that generically $\operatorname{dim}|-K_{X_{\boldsymbol{b}}}|=0$, with unique effective anticanonical divisor of the same type as $D$.
In this example we will arrive at a family of rational varieties which admits an action of $\widetilde{W}(A_1^{(1)}) \times \widetilde{W}(A_1^{(1)})$ by pseudo-isomorphisms, which are described on the level of the N\'eron-Severi bilattice in terms of an appropriate notion of Cremona isometry.
Then from a Cremona isometry corresponding to the pushforward $\tilde{\varphi}_*$ from the mapping \eqref{Map1} we obtain a deautonomisation of $\varphi$, and from another we obtain a different four-dimensional $q$-Painlev\'e-type equation.

\subsection{Cremona isometries}

Let us start with the following important definition, along the lines of that given in \cite{stefan} in the context of higher-dimensional analogues of discrete Painlev\'e equations, adapted from the one in \cite{sakai} (see also \cite{dolgachevweylgroups, Dolgachev-Ortland-1989}), 

\begin{definition}  \label{def:cremonaisometry}
Let $X$ be a rational variety with a chosen $D \in |-K_X|$. 
An automorphism $\sigma$ of the N\'eron-Severi bilattice $N(X)=H^2(X,\Z)\times H_2(X,\Z)$ (that is, a pair of $\Z$-module automorphisms of $H^2(X,\Z)$ and $H_2(X,\Z)$) is called a Cremona isometry if the following three properties hold:
\begin{enumerate}
    \item $\sigma$ preserves the intersection form;
    \item $\sigma$ on $H^2(X,\Z)$ preserves the decomposition of $D$ into irreducible components;
    \item $\sigma$ on $H^2(X,\Z)$ leaves the semigroup of classes of effective divisors invariant.
\end{enumerate}
The group of Cremona isometries of $X$ as above is denoted $\operatorname{Cr}(X)$.
\end{definition}

Note that any pseudo-automorphism of $X$ induces a Cremona isometry via pushforward or pullback. 
In the two-dimensional case, for a surface $S$ of the type associated with discrete Painlev\'e equations (generalised Halphen surfaces with zero-dimensional anticanonical linear system) Sakai \cite{sakai} described $\operatorname{Cr}(S)$ in terms of an affine Weyl group.
This group acts by reflections associated to a root system in the orthogonal complement of the components of $D$ in $\Pic(S)$.

In the higher-dimensional case, rather than roots in a single lattice, we want to describe Cremona isometries in terms of simple reflections $r_i$ defined by roots $\alpha_i \in H^2(X,\Z)$ and coroots $\alpha^{\vee}_i \in  H_2(X,\Z)$ according to
\begin{equation} \label{eq:refformula}
r_i (F) = F + \langle F , \alpha_i^{\vee} \rangle \alpha_i, ~~F \in H^2(X,\Z) \qquad r_i (f) = f + \langle \alpha_i , f \rangle \alpha^{\vee}_i, ~~f \in H_2(X,\Z).
\end{equation}
In the two-dimensional case, roots can be found by computing the orthogonal complement of the span of components of $D$ in $\Pic(S)$, but finding simple roots and coroots is more difficult in higher dimensions, see the discussion in \cite{stefan}. 
The following can be verified by direct calculation.

\begin{proposition} \label{prop:rootandcorootbases}
    For $X$ constructed in Section \ref{sec:regularisation}, we have bases of simple roots in $H^2(X,\Z)$ and simple coroots in $H_2(X,\Z)$ which form a root sytem of type $(\underset{ \left< \alpha, \alpha^\vee \right> = 14}{A_1 }+ A_1)^{(1)}$.
    Explicitly, the simple roots are
\begin{equation*}
\begin{aligned}
\alpha_0 &= 2 H_{q_1} + 2 H_{q_2} - E_{1,2,3,4,6,8,9,12,15,16,17,18,20,22,23,26} - 2 E_{14,28} + E_{7,21}, \\
\alpha_1 &= 2 H_{p_1} + 2 H_{p_2} - E_{1,2,5,8,9,10,11,13,15,16,19,22,23,24,25,27} - 2 E_{7,21} + E_{14,28}, \\
\alpha_2 &= H_{q_1} + H_{p_1} + H_{q_2} + H_{p_2} - E_{1,2,4,5,7,8,9,11,12,14,15,16,17,20,22,23,24,27}, \\
\alpha_3 &= H_{q_1} + H_{p_1} + H_{q_2} + H_{p_2} - E_{1,2,3,6,8,9,10,13,15,16,18,19,21,22,23,25,26,28},
\end{aligned} 
\end{equation*}
and the simple coroots are
\begin{equation*}
\begin{aligned}
\alpha^{\vee}_0 &= 2 h_{p_1} + 2 h_{p_2} - e_{2,3,16,17} - 2 e_{14,28} +e_{7,21}, \\
\alpha^{\vee}_1 &= 2 h_{q_1} + 2 h_{q_2} - e_{9,10,23,24} - 2 e_{7,21} +e_{14,28},\\
\alpha_2^{\vee} &= h_{q_1} + h_{p_1} + h_{q_2} + h_{p_2} - e_{7,14,16,17,23,24}, \\
\alpha_3^{\vee} &= h_{q_1} + h_{p_1} + h_{q_2} + h_{p_2} - e_{2,3,9,10,21,28}.
\end{aligned}
\end{equation*}
The simple coroots are orthogonal to the components of $D$, i.e. 
\begin{equation*}
    \langle D_i,\alpha_j^{\vee}\rangle=0\quad \text{ for } \quad  i=1,\dots,28,\, j=0,\dots,3,
\end{equation*}
and the anticanonical divisor class can be expressed as 
\begin{equation} \label{eq:nullroots}
    -K_X = \alpha_0+\alpha_1=\alpha_2+\alpha_3.
\end{equation}

\end{proposition}
The Cartan matrix associated to the root system in Proposition \ref{prop:rootandcorootbases} is
\begin{equation*}
\langle \alpha_i , \alpha^{\vee}_j \rangle = \left(\begin{array}{cccc}-14 & 14 & 0 & 0 \\14 & -14 & 0 & 0 \\0 & 0 & -2 & 2 \\0 & 0 & 2 & -2\end{array}\right).
\end{equation*}
The associated affine Weyl group, extended by Dynkin diagram automorphisms $\pi$ and $\sigma$, is
\begin{equation} \label{eq:weylgroup}
\widetilde{W}( {A_1^{(1)} } ) \times \widetilde{W}(A_1^{(1)}) 
= \langle r_0, r_1, \pi \rangle \times \langle r_2, r_3, \sigma \rangle,
\end{equation}
where the generators satisfy the relations
\begin{equation*}
    \begin{gathered}
    r_0^2 = r_1^2= \pi^2 =1, \quad r_1 \,\pi = \pi \,r_0, \qquad r_2^2 = r_3^2= \sigma^2 =1, \quad r_2 \,\sigma = \sigma \,r_3,  \qquad  \pi \, \sigma = \sigma \,\pi\\
    r_i r_j = r_j r_i , \quad r_i \sigma = \sigma r_i, \quad r_j \pi = \pi r_j, \text{ for } ~i \in \{0,1\}, ~j\in \{2,3\}.
    \end{gathered}
\end{equation*}
These are automorphisms of $N(X)$, which for $r_2$ and $r_3$ are given by formula \eqref{eq:refformula} with $\alpha_2,\alpha_3, \alpha_2^{\vee}, \alpha_3^{\vee}$ as in Proposition \ref{prop:rootandcorootbases}.
We collect the actions of the other generators on $H^2(X,\Z)$ in Figure \ref{fig:refsactiononcohom}.
\begin{figure}[htb]
\begin{equation*} 
\begin{aligned}
    &r_0 : \left\{ 
    \begin{aligned}
        H_{q_1} &\mapsto H_{q_2}, \quad H_{p_1} \mapsto H_{p_2}+H_{q_1}+H_{q_2}-E_{1,2,5,15,16,18,20,22,23,26,28}, \\
        H_{q_2} &\mapsto H_{q_1}, \quad H_{p_2} \mapsto H_{p_1}+H_{q_1}+H_{q_2}-E_{1,2,4,6,8,9,12,14,15,16,19},\\
        E_1 &\mapsto E_{18}, \quad 
        E_2 \mapsto H_{q_2}-E_{15,16,18,19,20}, \quad 
        E_3 \mapsto E_{21}, \quad 
        E_4 \mapsto E_{15},\\
        E_5 &\mapsto E_{20}, \quad 
        E_6 \mapsto E_{19}, \quad 
        E_7 \mapsto E_{17}, \quad 
        E_8 \mapsto E_{13}, \quad 
        E_9 \mapsto E_9, \\
        E_{10} &\mapsto E_{10}, \quad 
        E_{11} \mapsto E_{11}, \quad 
        E_{12} \mapsto H_{q_1}-E_{9,13,15,18,22,23,26}, \quad 
        E_{13} \mapsto E_8, \\ 
        E_{14} &\mapsto H_{q_2}-E_{28}, \quad 
        E_{15} \mapsto E_4, \quad 
        E_{16} \mapsto H_{q_1}-E_{1,2,4,5,6}, \quad 
        E_{17} \mapsto E_7,\\
        E_{18} &\mapsto E_1, \quad 
        E_{19} \mapsto E_6, \quad 
        E_{20} \mapsto E_5, \quad 
        E_{21} \mapsto E_3, \\
        E_{22} &\mapsto E_{27},\quad 
        E_{23} \mapsto E_{23}, \quad 
        E_{24} \mapsto E_{24}, \quad 
        E_{25} \mapsto E_{25},\\
        E_{26} &\mapsto H_{q_2}-E_{1,4,8,9,12,23,27}, \quad 
        E_{27} \mapsto E_{22}, \quad 
        E_{28} \mapsto H_{q_1}-E_{14},
    \end{aligned}
    \right.
    \\
    &r_1 : \left\{ 
    \begin{aligned}
    H_{q_1}&\mapsto H_{p_1}+H_{p_2}+H_{q_2}-E_{8,9,12,15,16,19,21,22,23,25,27}, \quad 
    H_{p_1} \mapsto H_{p_2},\\
    H_{q_2} &\mapsto H_{p_1}+H_{p_2}+H_{q_1}-E_{1,2,5,7,8,9,11,13,22,23,26}, \quad H_{p_2}\mapsto H_{p_1},\\
    E_1&\mapsto E_6, \quad 
    E_2\mapsto E_2, \quad 
    E_3\mapsto E_3,\quad 
    E_4 \mapsto E_4, \quad
    E_5 \mapsto H_{p_1}-E_{2,6,15,16,19,22,25}, \\
    E_6&\mapsto E_1, \quad
    E_7\mapsto H_{p_2}-E_{21}, \quad
    E_8\mapsto E_{25}, \quad 
    E_9\mapsto H_{p_2}-E_{22,23,25,26,27}, \\
    E_{10}&\mapsto E_{28}, \quad 
    E_{11}\mapsto E_{22}, \quad 
    E_{12}\mapsto E_{27}, \quad 
    E_{13}\mapsto E_{26}, \quad 
    E_{14}\mapsto E_{24}, \quad 
    E_{15}\mapsto E_{20}, \\
    E_{16}&\mapsto E_{16}, \quad
    E_{17}\mapsto E_{17}, \quad 
    E_{18}\mapsto E_{18}, \quad 
    E_{19}\mapsto H_{p_2}-E_{1,2,5,8,11,16,20}, \\
    E_{20}&\mapsto E_{15}, \quad 
    E_{21}\mapsto H_{p_1}-E_7, \quad
    E_{22}\mapsto E_{11}, \quad 
    E_{23}\mapsto H_{p_1}-E_{8,9,11,12,13},\\
    E_{24}&\mapsto E_{14}, \quad 
    E_{25}\mapsto E_8, \quad 
    E_{26}\mapsto E_{13}, \quad 
    E_{27}\mapsto E_{12}, \quad 
    E_{28}\mapsto E_{10},
    \end{aligned}
    \right.\\
    & \pi : \left\{
    \begin{aligned}
        H_{q_1}&\mapsto H_{p_1}, \quad 
        H_{p_1}\mapsto H_{q_1}, \quad 
        H_{q_2}\mapsto H_{p_2}, \quad 
        H_{p_2}\mapsto H_{q_2}, \\
        E_1&\mapsto E_8, \quad 
        E_2\mapsto E_9, \quad 
        E_3\mapsto E_{10}, \quad 
        E_4\mapsto E_{11}, \quad 
        E_5\mapsto E_{12}, \\ 
        E_6&\mapsto E_{13}, \quad 
        E_7\mapsto E_{14}, \quad
        E_8\mapsto E_1, \quad 
        E_9\mapsto E_2, \quad
        E_{10} \mapsto E_3,\\ 
        E_{11}&\mapsto E_4, \quad 
        E_{12}\mapsto E_5, \quad 
        E_{13}\mapsto E_6, \quad
        E_{14}\mapsto E_7, \quad
        E_{15}\mapsto E_{22}, \\
        E_{16}&\mapsto E_{23}, \quad 
        E_{17}\mapsto E_{24}, \quad 
        E_{18}\mapsto E_{25}, \quad
        E_{19}\mapsto E_{26}, \quad
        E_{20}\mapsto E_{27}, \\ 
        E_{21}&\mapsto E_{28}, \quad
        E_{22}\mapsto E_{15}, \quad
        E_{23}\mapsto E_{16}, \quad
        E_{24}\mapsto E_{17}, \\
        E_{25}&\mapsto E_{18},\quad
        E_{26}\mapsto E_{19}, \quad
        E_{27}\mapsto E_{20}, \quad
        E_{28}\mapsto E_{21},
    \end{aligned}
    \right. \\
    &\sigma : \left\{
    \begin{aligned}
        H_{q_1}&\mapsto H_{q_2}, \quad 
        H_{p_1}\mapsto H_{p_2}, \quad 
        H_{q_2}\mapsto H_{q_1}, \quad 
        H_{p_2}\mapsto H_{p_1}, \\
        E_1&\mapsto E_{15}, \quad 
        E_2\mapsto E_{16}, \quad
        E_3\mapsto E_{17}, \quad
        E_4\mapsto E_{18}, \quad
        E_5\mapsto E_{19}, \\
        E_6&\mapsto E_{20}, \quad
        E_7\mapsto E_{21}, \quad
        E_8\mapsto E_{22}, \quad
        E_9\mapsto E_{23}, \quad
        E_{10}\mapsto E_{24}, \\
        E_{11}&\mapsto E_{25}, \quad
        E_{12}\mapsto E_{26}, \quad
        E_{13}\mapsto E_{27}, \quad
        E_{14}\mapsto E_{28}, \quad
        E_{15}\mapsto E_1, \\
        E_{16}&\mapsto E_2, \quad
        E_{17}\mapsto E_3, \quad 
        E_{18}\mapsto E_4, \quad
        E_{19}\mapsto E_5,  \quad
        E_{20}\mapsto E_6, \\ 
        E_{21}&\mapsto E_7, \quad
        E_{22}\mapsto E_8, \quad
        E_{23}\mapsto E_9, \quad 
        E_{24}\mapsto E_{10}, \\ 
        E_{25}&\mapsto E_{11}, \quad
        E_{26}\mapsto E_{12}, \quad
        E_{27}\mapsto E_{13}, \quad
        E_{28}\mapsto E_{14}.
    \end{aligned}
    \right.
    \end{aligned}
\end{equation*}
    \caption{Action of $r_0$, $r_1$, $\pi$, $\sigma$ on $H^2(X,\Z)$}
    \label{fig:refsactiononcohom}
\end{figure}

In particular they act on the simple roots in Proposition \ref{prop:rootandcorootbases} in the usual way that affine Weyl groups act on the associated root lattices by formula \ref{eq:refformula}, namely
\begin{equation*}
\begin{gathered}
r_0 : \left\{ 
\begin{aligned}
\alpha_0 &\mapsto -\alpha_0, \\
\alpha_1 &\mapsto \alpha_1 + 2 \alpha_0, \\
\alpha_2 &\mapsto \alpha_2, \\
\alpha_3 &\mapsto \alpha_3, 
\end{aligned}
\right. 
\qquad 
r_1 : \left\{ 
\begin{aligned}
\alpha_0 &\mapsto \alpha_0 + 2 \alpha_1, \\
\alpha_1 &\mapsto -\alpha_1, \\
\alpha_2 &\mapsto \alpha_2 , \\
\alpha_3 &\mapsto \alpha_3, 
\end{aligned}
\right. 
\qquad 
\pi : \left\{ 
\begin{aligned}
\alpha_0 &\mapsto \alpha_1, \\
\alpha_1 &\mapsto \alpha_0, \\
\alpha_2 &\mapsto \alpha_2, \\
\alpha_3 &\mapsto \alpha_3.
\end{aligned}
\right. 
\\
r_2 : \left\{ 
\begin{aligned}
\alpha_0 &\mapsto \alpha_0, \\
\alpha_1 &\mapsto \alpha_1, \\
\alpha_2 &\mapsto - \alpha_2, \\
\alpha_3 &\mapsto \alpha_3 + 2 \alpha_2, 
\end{aligned}
\right. 
\qquad 
r_3 : \left\{ 
\begin{aligned}
\alpha_0 &\mapsto \alpha_0, \\
\alpha_1 &\mapsto \alpha_1, \\
\alpha_2 &\mapsto \alpha_2 + 2 \alpha_3, \\
\alpha_3 &\mapsto -\alpha_3, 
\end{aligned}
\right. 
\qquad 
\sigma : \left\{ 
\begin{aligned}
\alpha_0 &\mapsto \alpha_0, \\
\alpha_1 &\mapsto \alpha_1, \\
\alpha_2 &\mapsto \alpha_3, \\
\alpha_3 &\mapsto \alpha_2.
\end{aligned}
\right. 
\end{gathered}
\end{equation*}

The actions on $H_2(X,\Z)$ can be computed using the formula \eqref{H2H2}. 
\begin{remark}
    The non-standard normalisation $\langle \alpha,\alpha^{\vee}\rangle = 14$ for one of the copies of the $A_1^{(1)}$ root system is reminiscent of the symmetry types of surfaces in the Sakai classification with non-standard root lengths, e.g. surface type $R=A_6^{(1)}$ with symmetry type $R^{\perp} = (\underset{{\tiny |\alpha|^2= 14}}{A_1}+ A_1)^{(1)}$ in the notation of \cite{sakai}.
    Note that the action on $N(X)$ of the copy of $\widetilde{W}( A_1^{(1)} )$ corresponding to the $\underset{ \left< \alpha, \alpha^\vee \right> = 14}{A_1^{(1)}}$ root system is not directly by the formula \eqref{eq:refformula}. 
    It acts on the roots and coroots in the usual way, but not by the formula \eqref{eq:refformula} on the whole of $N(X)$ and this is why we provided the actions of $r_0, r_1$ explicitly in Figure \ref{fig:refsactiononcohom}.
    This is analogous to what happens in two-dimensional cases with non-standard root lengths, since reflections there will permute components of $D$.
\end{remark}
The proof that the elements of $\widetilde{W}( {A_1^{(1)} } ) \times \widetilde{W}(A_1^{(1)})$ provide Cremona isometries will be deferred to the next subsection when we realise them by pseudo-isomorphisms.
The reason for this is that, while conditions (1) and (2) in Definition \ref{def:cremonaisometry} are immediate, to establish that effectiveness is preserved requires more work. In the two-dimensional case this can be done using classical surface theory \cite[Sect. 6]{sakai}, but in higher dimensions it is more practical to realise them as pseudo-isomorphisms to show they preserve effectiveness.

\subsection{Cremona action by pseudo-isomorphisms}

We will realise the group \eqref{eq:weylgroup} as an action by pseudo-isomorphisms.
In parallel with the two-dimensional case \cite{fane1}, we will allow centres of blow-ups from the construction of $X$ to move, while keeping a choice of anticanonical divisor $D=\sum_{i}m_iD_i\in |-K_X|$ intact. 
The choice of effective anticanonical divisor we will choose to preserve is that given in \eqref{decomp}.
This will lead to a family of varieties
$$X_{\boldsymbol{b}},  \qquad \boldsymbol{b} = (b_3, b_7, b_{10}, b_{14}, b_{17}, b_{21}, b_{24}, b_{28}) \in \C^8,$$ 
containing $X$ and such that $\operatorname{dim}|-K_{X_{\boldsymbol{b}}}|=0$ for generic $\boldsymbol{b}$, with the unique effective anticanonical divisor of $X_{\boldsymbol{b}}$ being $D=\sum_i D_i$ as in \eqref{decomp}. 
The parameters introduced here are indexed according to which $V_i$ they allow to move. 
We give the centres of blow-ups to obtain $X_{\boldsymbol{b}}$ from $(\p^1)^{\times 4}$
explicitly in Figures \ref{fig:firstsevenblow-upsNA} -- \ref{fig:fourthsevenblow-upsNA}.
Note that we could use the $\p \operatorname{G L}(2,\C)^{\times 4}$ automorphism group of $(\p^1)^{\times 4}$ to normalise some of the parameters by change of coordinates. 
We choose not to do this at this stage, to leave as much parameter freedom as possible so it is easier to obtain actions on the family by pseudo-isomorphisms, which can be used to derive $q$-difference equations.


\begin{figure}[h]
    \centering
    \begin{tikzpicture}    
 \node at (0,0)
 {
 \begin{tikzcd}
    [column sep=-0.25in]
        &       &V_1 : \left(\begin{array}{c} 0 \\ * \\ \infty \\ \infty  \end{array}\right)   &       & V_2 : \left(\begin{array}{c} 0 \\ \infty \\ * \\ \infty  \end{array}\right)  &  & \quad \arrow[ll, phantom, "\leftarrow"] & \\
        &V_4 : \left(\begin{array}{c} 0 \\ * \\ \infty \\ *  \end{array}\right) \arrow[ur, phantom,sloped, "\supset"]  &       &V_5 : \left(\begin{array}{c} 0 \\ * \\ * \\ \infty  \end{array}\right) \arrow[ul, phantom,sloped, "\subset"] \arrow[ur, phantom,sloped, "\supset"] &       &V_6 : \left(\begin{array}{c} 0 \\ \infty \\ * \\ *  \end{array}\right)\arrow[ul, phantom,sloped, "\subset"] &  \qquad \qquad & V_7: \left(\begin{array}{c} 0 \\ b_{7} \\ * \\ * \end{array}\right)
    \end{tikzcd}
    };
    \node at (5,1.4) 
    {
    $V_3 : \left(\begin{array}{c} q_2 - \frac{b_3}{q_1 p_1}  \\ \frac{1}{q_1 p_2} \\ q_1 \\ q_2  \end{array}\right) = \left(\begin{array}{c} 0 \\ * \\ 0 \\ * \end{array}\right)$
    };
\end{tikzpicture}
    \caption{The first seven blow-ups in the construction of $X_{\boldsymbol{b}}$.}
    \label{fig:firstsevenblow-upsNA}
\end{figure}


\begin{figure}[h]
    \centering
    \begin{tikzpicture}    
 \node at (0,0)
 {
 \begin{tikzcd}
    [column sep=-0.25in]
        &       &V_8 : \left(\begin{array}{c} * \\ 0 \\ \infty \\ \infty  \end{array}\right)   &       & V_{9} : \left(\begin{array}{c} \infty \\ 0 \\ \infty \\ *  \end{array}\right)   &  & \quad \arrow[ll, phantom, "\leftarrow"] & \\
        &V_{11} : \left(\begin{array}{c} * \\ 0 \\ * \\ \infty  \end{array}\right) \arrow[ur, phantom,sloped, "\supset"]  &       &V_{12} : \left(\begin{array}{c} * \\ 0 \\ \infty \\ *  \end{array}\right) \arrow[ul, phantom,sloped, "\subset"] \arrow[ur, phantom,sloped, "\supset"] &       &V_{13} : \left(\begin{array}{c} \infty \\ 0 \\ * \\ *  \end{array}\right)\arrow[ul, phantom,sloped, "\subset"] &  \qquad \qquad & V_{14} : \left(\begin{array}{c} b_{14} \\ 0 \\ * \\ * \end{array}\right)
    \end{tikzcd}
    };
    \node at (5,1.4) 
    {
    $V_{10} : \left(\begin{array}{c} p_2 - \frac{b_{10}}{q_1 p_1}  \\ \frac{1}{p_1 q_2} \\ p_1 \\ p_2  \end{array}\right) = \left(\begin{array}{c} 0 \\ * \\ 0 \\ * \end{array}\right)$
    };
\end{tikzpicture}
    \caption{The second seven blow-ups in the construction of $X_{\boldsymbol{b}}$.}
    \label{fig:secondsevenblow-upsNA}
\end{figure}


\begin{figure}[h]
    \centering
    \begin{tikzpicture}    
 \node at (0,0)
 {
 \begin{tikzcd}
    [column sep=-0.25in]
        &       &V_{15} : \left(\begin{array}{c} \infty \\ \infty \\ 0 \\ *  \end{array}\right)   &       & V_{16} : \left(\begin{array}{c} * \\ \infty \\ 0 \\ \infty  \end{array}\right)   &  & \quad \arrow[ll, phantom, "\leftarrow"] & \\
        &V_{18} : \left(\begin{array}{c} \infty \\ * \\ 0 \\ *  \end{array}\right) \arrow[ur, phantom,sloped, "\supset"]  &       &V_{19} : \left(\begin{array}{c} * \\ \infty \\ 0 \\ *  \end{array}\right) \arrow[ul, phantom,sloped, "\subset"] \arrow[ur, phantom,sloped, "\supset"] &       &V_{20} : \left(\begin{array}{c} * \\ * \\ 0 \\ \infty  \end{array}\right)\arrow[ul, phantom,sloped, "\subset"] &  \qquad \qquad & V_{21} : \left(\begin{array}{c}  * \\ * \\ 0 \\ b_{21}\end{array}\right)
    \end{tikzcd}
    };
    \node at (5.1,1.4) 
    {
    $V_{17} : \left(\begin{array}{c} q_1 - \frac{b_{17}}{q_2 p_2}  \\ \frac{1}{p_1 q_2} \\ q_2 \\ q_1  \end{array}\right) = \left(\begin{array}{c} 0 \\ * \\ 0 \\ * \end{array}\right)$
    };
\end{tikzpicture}
    \caption{The third seven blow-ups in the construction of $X_{\boldsymbol{b}}$.}
    \label{fig:thirdsevenblow-upsNA}
\end{figure}


\begin{figure}[h]
    \centering
    \begin{tikzpicture}    
 \node at (0,0)
 {
 \begin{tikzcd}
    [column sep=-0.25in]
        &       &V_{22} : \left(\begin{array}{c} \infty \\ \infty \\ * \\ 0  \end{array}\right)   &       & V_{23} : \left(\begin{array}{c} \infty \\ * \\ \infty \\ 0  \end{array}\right)   &  & \quad \arrow[ll, phantom, "\leftarrow"] & \\
        &V_{25} : \left(\begin{array}{c} * \\ \infty \\ * \\ 0  \end{array}\right) \arrow[ur, phantom,sloped, "\supset"]  &       &V_{26} : \left(\begin{array}{c} \infty \\ * \\ * \\ 0  \end{array}\right) \arrow[ul, phantom,sloped, "\subset"] \arrow[ur, phantom,sloped, "\supset"] &       &V_{27} : \left(\begin{array}{c} * \\ * \\ \infty \\ 0 \end{array}\right)\arrow[ul, phantom,sloped, "\subset"] &  \qquad \qquad & V_{28} : \left(\begin{array}{c}  * \\ * \\  b_{28} \\ 0\end{array}\right)
    \end{tikzcd}
    };
    \node at (5,1.4) 
    {
    $V_{24} : \left(\begin{array}{c} p_1 - \frac{b_{24}}{q_2 p_2}  \\ \frac{1}{q_1 p_2} \\ p_2 \\ p_1  \end{array}\right) = \left(\begin{array}{c} 0 \\ * \\ 0 \\ * \end{array}\right)$
    };
\end{tikzpicture}
    \caption{The fourth seven blow-ups in the construction of $X_{\boldsymbol{b}}$.}
    \label{fig:fourthsevenblow-upsNA}
\end{figure}

For each $w \in \widetilde{W}(A_1^{(1)}) \times \widetilde{W}(A_1^{(1)})$ we construct an action on the family of $X_{\boldsymbol{b}}$ by pseudo-isomorphisms  which induces $w$ as a Cremona isometry by pushforward.


\begin{remark}
    Note that, while a Cremona isometry is defined with reference to only a single variety $X$, the enumeration of blowups in the construction of $X_{\boldsymbol{b}}$ means that there is a natural identification of $N(X_{\boldsymbol{b}})$ and $N(X_{\boldsymbol{b}'})$ with $N(X)$, and we can associate a Cremona isometry to any bilattice isomorphism $N(X_{\boldsymbol{b}})\to N(X_{\boldsymbol{b}'})$ with properties corresponding to those in Definition \ref{def:cremonaisometry}.
    This is a technical detail which we will sometimes abuse notation to avoid.
\end{remark}

\begin{theorem} \label{th:cremonaction}
    We have a Cremona action on the family of varieties $X_{\boldsymbol{b}}$ by pseudo-isomorphisms which realises the Cremona isometries above.
    That is, for each $w\in \widetilde{W}(A_1^{(1)})\times\widetilde{W}(A_1^{(1)})$ acting on $N(X)$ as above, we have an action on parameters $\boldsymbol{b}\mapsto\boldsymbol{b}'$ and a pseudo-isomorphism $X_{\boldsymbol{b}}\to X_{\boldsymbol{b}'}$, such that its pushforward on $N(X)$ coincides with $w$.

    Below we give the actions on parameters, as well as birational self-maps of $(\p^1)^{\times 4}$ in affine coordinates $(q_1,p_1,q_2,p_2)$, which when lifted under the blow-ups give the required pseudo-isomorphisms.
     \begin{equation*}
        r_0 : \left\{ 
            \begin{aligned}
            q_1 &\mapsto q_2, \\
            p_1 &\mapsto \frac{1 - b_{28}^{-1} q_2}{q_1 q_2 p_2},\\
            q_2 &\mapsto q_1, \\
            p_2 &\mapsto \frac{1- b_{14}^{-1} q_1}{q_1 p_1 q_2},    
            \end{aligned}
            \qquad 
            \begin{aligned}
            b_3 &\mapsto b_{21}^{-1}, \\
            b_{7} &\mapsto b_{17}^{-1}, \\
            b_{10} &\mapsto b_{10}^{-1}b_{14}^{-1}b_{28}^{-1},  \\
            b_{14} &\mapsto b_{28}, 
            \end{aligned}
            \qquad 
            \begin{aligned}
            b_{17} &\mapsto b_{7}^{-1}, \\
            b_{21} &\mapsto b_3^{-1}, \\
            b_{24} &\mapsto b_{14}^{-1}b_{24}^{-1}b_{28}^{-1}, \\
            b_{28} &\mapsto b_{14}.
            \end{aligned}
        \right. 
    \end{equation*}

    \begin{equation*}
        r_1 : \left\{ 
            \begin{aligned}
            q_1 &\mapsto \frac{1 - b_{21}^{-1} p_2}{p_1 q_2 p_2}, \\
            p_1 &\mapsto p_2,\\
            q_2 &\mapsto \frac{1 - b_{7}^{-1} p_2}{q_1 p_1 p_2}, \\
            p_2 &\mapsto p_1,    
            \end{aligned}
            \qquad 
            \begin{aligned}
            b_3 &\mapsto b_3^{-1}b_7^{-1}b_{21}^{-1}, \\
            b_{7} &\mapsto b_{21}, \\
            b_{10} &\mapsto b_{28}^{-1},  \\
            b_{14} &\mapsto b_{24}^{-1}, 
            \end{aligned}
            \qquad 
            \begin{aligned}
            b_{17} &\mapsto b_{7}^{-1}b_{17}^{-1}b_{21}^{-1}, \\
            b_{21} &\mapsto b_{7}, \\
            b_{24} &\mapsto b_{14}^{-1}, \\
            b_{28} &\mapsto b_{10}^{-1}.
            \end{aligned}
        \right. 
    \end{equation*}
    
    \begin{equation*}
        \pi : \left\{ 
            \begin{aligned}
            q_1 &\mapsto p_1, \\
            p_1 &\mapsto q_1, \\
            q_2 &\mapsto p_2, \\
            p_2 &\mapsto q_2,
            \end{aligned}
        \qquad 
            \begin{aligned}
            b_{3} &\mapsto b_{10}, \\
            b_{7} &\mapsto b_{14}, \\
            b_{10} &\mapsto b_{3},  \\
            b_{14} &\mapsto b_{7}, 
            \end{aligned}
            \qquad 
            \begin{aligned}
            b_{17} &\mapsto b_{24}, \\
            b_{21} &\mapsto b_{28}, \\
            b_{24} &\mapsto b_{17}, \\
            b_{28} &\mapsto b_{21}.
            \end{aligned}
        \right. 
    \end{equation*}

    \begin{equation*}
        r_2 : \left\{ 
            \begin{aligned}
            q_1 &\mapsto 
            c q_1\frac{ q_1 p_1 q_2 p_2 -b_{24} q_1 -b_{17} p_1 +  b_{14} b_{24}}{q_1 p_1 q_2 p_2 - b_7 b_{14}^{-1}b_{17}  q_1- b_{17} p_1+  b_7 b_{17}}, \\
            p_1 &\mapsto 
            \frac{p_1}{c}\frac{q_1 p_1 q_2 p_2-b_{24} q_1 -b_{17} p_1 + b_{17} b_7}{ q_1 p_1 q_2 p_2 -b_{24}  q_1 - b_{7}^{-1}b_{14} b_{24} p_1+b_{14} b_{24}},\\
            q_2 &\mapsto 
            c q_2 \frac{ q_1 p_1 q_2 p_2 -b_{24}  q_1 - b_{7}^{-1} b_{14} b_{24} p_1 +b_{14} b_{24}}{ q_1 p_1 q_2 p_2  - b_7 b_{14}^{-1}b_{17}  q_1- b_{17} p_1 + b_7 b_{17} }, \\
            p_2 &\mapsto 
            \frac{p_2}{c}\frac{q_1 p_1 q_2 p_2 - b_7 b_{14}^{-1}b_{17}  q_1 - b_{17} p_1 + b_7 b_{17} }{  q_1 p_1 q_2 p_2 -b_{24}  q_1  -b_7^{-1} b_{14} b_{24} p_1+ b_{14} b_{24}},    
            \end{aligned}
            \qquad 
            \begin{aligned}
            b_3 &\mapsto c b_3, \\
            b_{7} &\mapsto c^{-1} b_{7}, \\
            b_{10} &\mapsto c^{-1} b_{10},  \\
            b_{14} &\mapsto c b_{14}, 
            \end{aligned}
            \qquad 
            \begin{aligned}
            b_{17} &\mapsto c^{-1} b_{17}, \\
            b_{21} &\mapsto c b_{21}, \\
            b_{24} &\mapsto c b_{24}, \\
            b_{28} &\mapsto c^{-1} b_{28},
            \end{aligned}
        \right. 
    \end{equation*}
    where $c^2 = b_{7}b_{17}b_{14}^{-1}b_{24}^{-1}$, and in particular $r_2 : c\mapsto c^{-1}$.

    \begin{equation*}
        r_3 : \left\{ 
            \begin{aligned}
            q_1 &\mapsto 
           \tilde{c} q_1 \frac{ q_1 p_1 q_2 p_2 - b_{10}  q_2-b_{10}b_{21}^{-1}b_{28} p_2+ b_{10}  b_{28} }{q_1 p_1 q_2 p_2 -b_{3}b_{21}b_{28}^{-1} q_2 - b_{3}p_2+  b_{3}b_{21}}, \\
            p_1 &\mapsto 
            \frac{p_1}{\tilde{c} } \frac{  q_1 p_1 q_2 p_2- b_{3} b_{21} b_{28}^{-1} q_2- b_{3} p_2+ b_{3} b_{21}  }{ q_1 p_1 q_2 p_2-b_{10}  q_2-b_{10}b_{21}^{-1} b_{28} p_2+ b_{10}  b_{28}},\\
            q_2 &\mapsto 
            \tilde{c} q_2 \frac{q_1 p_1 q_2 p_2-b_{10} q_2-b_{3} p_2 +  b_{10} b_{28}}{q_1 p_1 q_2 p_2 -  b_{3}b_{21} b_{28}^{-1} q_2- b_{3} p_2+ b_{3}b_{21}  }, \\
            p_2 &\mapsto 
            \frac{p_2}{\tilde{c}} \frac{ q_1 p_1 q_2 p_2 -b_{10} q_2-b_{3} p_2 +b_{3} b_{21} }{ q_1 p_1 q_2 p_2 -b_{10} q_2-b_{10} b_{21}^{-1}b_{28} p_2 + b_{10}  b_{28}},    
            \end{aligned}
            \qquad 
            \begin{aligned}
            b_3 &\mapsto \tilde{c}^{-1} b_{3}, \\
            b_{7} &\mapsto \tilde{c} b_{7} , \\
            b_{10} &\mapsto \tilde{c} b_{10},  \\
            b_{14} &\mapsto \tilde{c}^{-1} b_{14}, 
            \end{aligned}
            \qquad 
            \begin{aligned}
            b_{17} &\mapsto \tilde{c} b_{17}, \\
            b_{21} &\mapsto \tilde{c}^{-1} b_{21}, \\
            b_{24} &\mapsto \tilde{c}^{-1} b_{24}, \\
            b_{28} &\mapsto \tilde{c} b_{28},
            \end{aligned}
        \right. 
    \end{equation*}
    where $\tilde{c}^2 = b_{3}b_{21}b_{10}^{-1}b_{28}^{-1}$, and in particular $r_3 : \tilde{c}\mapsto \tilde{c}^{-1}$.
    \begin{equation*}
        \sigma : \left\{ 
            \begin{aligned}
            q_1 &\mapsto q_2, \\
            p_1 &\mapsto p_2,\\
            q_2 &\mapsto q_1, \\
            p_2 &\mapsto p_1,    
            \end{aligned}
            \qquad 
            \begin{aligned}
            b_3 &\mapsto b_{17}, \\
            b_{7} &\mapsto b_{21}, \\
            b_{10} &\mapsto b_{24},  \\
            b_{14} &\mapsto b_{28}, 
            \end{aligned}
            \qquad 
            \begin{aligned}
            b_{17} &\mapsto b_3, \\
            b_{21} &\mapsto b_{7}, \\
            b_{24} &\mapsto b_{10}, \\
            b_{28} &\mapsto b_{14}.
            \end{aligned}
        \right. 
    \end{equation*}
\end{theorem}

The proof of Theorem \ref{th:cremonaction} is by long but straightforward computations.
These are done in local charts for $X_{\boldsymbol{b}}$ introduced along the same lines as in Section \ref{subsec:blow-ups}, using the centres of blow-ups leading to $X_{\boldsymbol{b}}$ as given in Figure \ref{fig:firstsevenblow-upsNA}--\ref{fig:fourthsevenblow-upsNA}.

\subsection{4D Discrete Painlev\'e equations}

In the Sakai framework, discrete Painlev\'e equations arise from the Cremona action of translation elements of the extended affine Weyl group of Cremona isometries. 
To account for all non-autonomous mappings which are integrable in the sense of quadratic degree growth one must relax the translation requirement and allow elements of infinite order \cite{Mase}.
Famous early examples of discrete Painlev\'e equations often turn out to correspond to such non-translation elements, e.g. \cite{dP2}. 
Since the subgroup of translations is of finite index these elements of infinite order are always `quasi-translations', in the sense that they become translations when raised to some power. Non-autonomous mappings with `translational' parameter evolution (in an additive, multiplicative, or elliptic sense depending on the surface type) can still be obtained from quasi-translations via the process of projective reduction \cite{projectivereduction}.
These quasi-translations can often be understood (see e.g. \cite{yangtranslations}) as translations with respect to the affine Weyl group structure of some subgroup of symmetries compatible with the corresponding projective reduction condition, so one can still say that discrete Painlev\'e equations come from translations.

We now use the Cremona action in Theorem \ref{th:cremonaction} to construct fourth-order analogues of discrete Painlev\'e equations from translation elements.
We have two copies of $W(A_1^{(1)})$, each of which contains a subgroup of translations.
The translation part of the first copy $W(A_1^{(1)})=\langle r_0,r_1,\pi\rangle$ is generated by, e.g., $\pi \,r_0$, which acts on the roots by

\begin{equation*}
\pi \, r_0 : 
\left(\begin{array}{c}\alpha_0 \\ \alpha_1 \\ \alpha_2 \\ \alpha_3\end{array}\right)  \mapsto \left(\begin{array}{c}\alpha_0-\delta \\ \alpha_1+\delta\\ \alpha_2 \\ \alpha_3\end{array}\right), \qquad \delta = - K_X = \alpha_0 + \alpha_1.
\end{equation*}
\begin{proposition}
    The Cremona action of the translation $\pi \,r_0$
    gives  
    \begin{equation} \label{eq:4DqP1cremona1}
        \pi \,r_0 : \left\{ 
            \begin{aligned}
            q_1 &\mapsto \frac{1 - b_{28}^{-1} q_2}{q_1 q_2 p_2}, \\
            p_1 &\mapsto q_2,\\
            q_2 &\mapsto \frac{1 - b_{14}^{-1} q_1}{q_1 p_1 q_2}, \\
            p_2 &\mapsto q_1,    
            \end{aligned}
            \qquad 
            \begin{aligned}
            b_3 &\mapsto b_{10}^{-1}b_{14}^{-1}b_{28}^{-1}, \\
            b_{7} &\mapsto b_{28}, \\
            b_{10} &\mapsto b_{21}^{-1},  \\
            b_{14} &\mapsto b_{17}^{-1}, 
            \end{aligned}
            \qquad 
            \begin{aligned}
            b_{17} &\mapsto b_{14}^{-1}b_{24}^{-1}b_{28}^{-1}, \\
            b_{21} &\mapsto b_{14}, \\
            b_{24} &\mapsto b_{7}^{-1}, \\
            b_{28} &\mapsto b_{3}^{-1}.
            \end{aligned}
        \right. 
    \end{equation}
    This can be used to derive the following system of $q$-difference equations (in which we denote the shift parameter by $\lambda$ instead of $q$):
    \begin{equation} \label{eq:4DqP1}
\left\{
\begin{aligned}
\bar{q}_1 &= \frac{1 - a_1 t q_2}{q_1 q_2 p_2}, \\
\bar{p}_1 &= q_2, \\
\bar{q}_2 &= \frac{1 - a_2 t q_1}{q_1 p_1 q_2}, \\
\bar{p}_2 &= q_1,
\end{aligned}
\right.
    \end{equation}
where $a_1,a_2$ are parameters, $q_1=q_1(t)$, $p_1=p_1(t)$, $q_2=q_2(t)$, $p_2=p_2(t)$ and we denote $\lambda$-shift in $t$ by $\bar{q}_1=q_1( \lambda t)$, $\bar{p}_1=p_1(\lambda t)$, $\bar{q}_2=q_2( \lambda t)$, $\bar{p}_2=p_2( \lambda t)$.
\end{proposition}
We regard the system \eqref{eq:4DqP1} as a deautonomised version of the mapping \eqref{Map1}, since the action by pushforward of \eqref{eq:4DqP1cremona1} on the N\'eron-Severi bilattice coincides with that in equation \eqref{pushforwardcohom} induced by \eqref{Map1}.
\begin{remark}
    Though $\pi \,r_0$ is a translation element of $W(A_1^{(1)})$, and as such acts on the roots $\alpha_i$ as shifts by the null root $\delta$, both its action on $N(X)$ and the action on parameters $b_i$ are not purely (multiplicative) translational. 
    The former fact is related to the non-standard normalisation of the corresponding $A_1^{(1)}$ root system, and similar phenomena can also be observed in discrete Painlev\'e equations from surfaces in the classification whose symmetry types have non-standard root lengths.
    The latter fact, that the action on parameters is non-translational, is more subtle. 
    In the two-dimensional case, the period map construction ensures that surfaces can be parametrised by root variables, on which the Cremona action corresponds directly to the action on roots. 
    We do not yet have a geometric construction of root variables or an analogue of the Torelli-type theorem \cite[Th. 25]{sakai} in higher dimensions, so we cannot say how much of the non-translational action on the $b_i$'s can be removed while staying in the same isomorphism class of $X_{\boldsymbol{b}}$.
\end{remark}
    Nevertheless, one can derive \eqref{eq:4DqP1} by specialising parameters $b_i$ such that the action in \eqref{eq:4DqP1cremona1} becomes translational, in the same spirit as projective reduction.
    In the case at hand we see that the following conditions are sufficient for `translational' parameter evolution:
    \begin{equation}
        b_{10} b_{14} = b_{24} b_{28} = \frac{1}{b_{17} b_{21}} = \frac{1}{ b_{3} b_{7}}.
    \end{equation}
    So parametrising $b_i$ by $t$ according to
    \begin{equation}
        \begin{aligned}
            b_3 &= \frac{1}{a_1 \lambda t}, \\
            b_{7} &= \frac{a_1 t}{\lambda},
        \end{aligned}
        \qquad
        \begin{aligned}
            b_{10} &= \frac{\lambda^2}{a_2 t},  \\
            b_{14} &= a_2 t, 
        \end{aligned}
            \qquad 
        \begin{aligned}
            b_{17} &= \frac{1}{a_2 \lambda t}, \\
            b_{21} &= \frac{a_2 t}{\lambda}, 
        \end{aligned}
        \qquad
        \begin{aligned}
            b_{24} &= \frac{\lambda^2}{a_1 t}, \\
            b_{28} &= a_1 t,
        \end{aligned}
    \end{equation}
    the evolution of $b_i$ in \eqref{eq:4DqP1cremona1} is induced by $t\mapsto \lambda t$.

    
The translation part of the second copy $W(A_1^{(1)})=\langle r_2,r_3,\sigma\rangle$ is generated by, e.g., $\sigma \,r_2$, which acts on the roots by
\begin{equation*}
\sigma \, r_2 : \left(\begin{array}{c}\alpha_0 \\ \alpha_1 \\ \alpha_2 \\ \alpha_3\end{array}\right)  \mapsto \left(\begin{array}{c}\alpha_0 \\ \alpha_1 \\ \alpha_2-\delta \\ \alpha_3+\delta\end{array}\right), \qquad  \delta = - K_X = \alpha_2 + \alpha_3.
\end{equation*}

\begin{proposition}
    The Cremona action of the translation $\sigma \,r_2$ gives 
    \begin{equation} \label{eq:4DqP1cremona2}
        \sigma \,r_2 : \left\{ 
            \begin{aligned}
            q_1 &\mapsto 
            c q_2\frac{q_1 p_1 q_2 p_2-b_{24} q_1 - b_{7}^{-1}b_{14} b_{24} p_1 + b_{14} b_{24}}{ q_1 p_1 q_2 p_2 -  b_7 b_{14}^{-1}b_{17}  q_1 - b_{17} p_1+b_{7} b_{17}}, \\
            p_1 &\mapsto 
            \frac{p_2}{c}\frac{  q_1 p_1 q_2 p_2 -  b_7 b_{14}^{-1}b_{17}  q_1 - b_{17} p_1+b_{7} b_{17}}{q_1 p_1 q_2 p_2-b_{24} q_1 - b_{7}^{-1}b_{14} b_{24} p_1 + b_{14} b_{24}}
            ,\\
            q_2 &\mapsto c q_1 \frac{q_1 p_1 q_2 p_2 - b_{24}  q_1 - b_{17} p_1 + b_{14} b_{24} }{  q_1 p_1 q_2 p_2 - b_7 b_{14}^{-1}b_{17}  q_1  -b_{17} p_1+ b_{7} b_{17}}
            , \\
            p_2 &\mapsto 
            \frac{p_1}{c} \frac{ q_1 p_1 q_2 p_2 -b_{24}  q_1 - b_{17} p_1 +b_{7} b_{17}}{ q_1 p_1 q_2 p_2  - b_{24}  q_1- b_{7}^{-1}b_{14} b_{24} p_1 + b_{14} b_{24} }
            ,    
            \end{aligned}
            \qquad 
            \begin{aligned}
            b_3 &\mapsto c^{-1} b_{17}, \\
            b_{7} &\mapsto c b_{21}, \\
            b_{10} &\mapsto c b_{24},  \\
            b_{14} &\mapsto c^{-1} b_{28}, 
            \end{aligned}
            \qquad 
            \begin{aligned}
            b_{17} &\mapsto c b_3, \\
            b_{21} &\mapsto c^{-1} b_7, \\
            b_{24} &\mapsto c^{-1} b_{10}, \\
            b_{28} &\mapsto c b_{14},
            \end{aligned}
        \right. 
    \end{equation}
    where again $c^2 = b_{7}b_{17}b_{14}^{-1}b_{24}^{-1}$, and in particular $\sigma \,r_2 : c\mapsto c^2 \tilde{c}$, where again $\tilde{c}^2 = b_{3}b_{21}b_{10}b_{28}^{-1}$.
    
    This can be used to derive the following system of $q$-difference equations (in which we denote the shift parameter by $\mu$ instead of $q$) in which $\gamma_1,\gamma_2,\gamma_3,\gamma_4$ are constant parameters and $\bar{~}$ indicates $q$-shift $t\mapsto \mu t$, with $\mu = \sqrt{\frac{\gamma_1 \gamma_2}{\gamma_3 \gamma_4}}$:
\begin{equation} \label{eq:4DqP12}
\left\{
\begin{aligned}
\bar{q}_1 &=  \mu  q_2\frac{ t^2 q_1 p_1 q_2 p_2-\gamma_3 t q_1 - \mu^{-2} \gamma_1  t^{-1} p_1 + \gamma_3  \gamma_4}{ q_1 p_1 q_2 p_2 -   \mu^2 \gamma_3 t^3  q_1 - \gamma_1 t p_1+\gamma_1  \gamma_2 t^2}, \\
\bar{p}_1 &= \frac{p_2}{\mu }\frac{  q_1 p_1 q_2 p_2 -  \mu^2 \gamma_3   t^3  q_1 - \gamma_1 t p_1+\gamma_1 \gamma_2 t^2}{ t^2 q_1 p_1 q_2 p_2-\gamma_3 t q_1 - \mu^{-2} \gamma_1  t^{-1} p_1 + \gamma_3  \gamma_4 }, \\
\bar{q}_2 &= \mu  q_1 \frac{t^2 q_1 p_1 q_2 p_2 - \gamma_3 t  q_1 - \gamma_1 t^3 p_1 + \gamma_3  \gamma_4 }{  q_1 p_1 q_2 p_2 - \mu^2 \gamma_3   t^3  q_1  -\gamma_1 t p_1+ \gamma_1  \gamma_2 t^2}, \\
\bar{p}_2 &= \frac{p_1}{\mu} \frac{ q_1 p_1 q_2 p_2 -\gamma_3 t^{-1}  q_1 - \gamma_1 t p_1 +\gamma_1  \gamma_2 t^2}{ t^2 q_1 p_1 q_2 p_2  - \gamma_3 t  q_1- \mu^{-2} \gamma_1  t^{-1} p_1 + \gamma_3 \gamma_4 }.
\end{aligned}
\right. 
    \end{equation}

\end{proposition}

    Similarly to the derivation of the system of $q$-difference equations \eqref{eq:4DqP1} through appropriate parameter constraints such that \eqref{eq:4DqP1cremona1} becomes translational, for the Cremona action of $\sigma\,r_2$ we find the following sufficient conditions for translational parameter evolution:
    \begin{equation}
        \frac{b_{7}}{b_{21}} = \frac{b_{17}}{b_{3}}= \frac{b_{10}}{b_{24}}= \frac{b_{28}}{b_{14}} = t^2,
    \end{equation}
    Parametrising $b_i$ by $t$ according to
    \begin{equation}
        \begin{aligned}
            b_3 &= \frac{\gamma_{1}}{t}, &\quad &b_{10} = \gamma_{3} t,  &\quad &b_{17} = \gamma_{1} t  &\quad &b_{24}= \frac{\gamma_{3}}{t},\\
            b_7 &= \gamma_2 t, &\quad &b_{14} = \frac{\gamma_{4}}{t},  &\quad &b_{21} = \frac{\gamma_{2}}{t}  &\quad &b_{28}= \gamma_{4} t,
        \end{aligned}
    \end{equation}
    for constants $\gamma_1,\gamma_2,\gamma_3, \gamma_4$, so $c = \sqrt{\frac{\gamma_1\gamma_{2}}{\gamma_3 \gamma_4}}t^2$, means the evolution of $b_i$ in \eqref{eq:4DqP1cremona2} is induced by $t\mapsto \mu t$, with $\mu = \sqrt{\frac{\gamma_1\gamma_{2}}{\gamma_3 \gamma_4}}$.
    This gives the system of $q$-difference equations \eqref{eq:4DqP12}.

    \begin{remark}
        It is straightforward to check from the action of $\sigma\, r_2$ on $N(X)$ that the system \eqref{eq:4DqP12} exhibits quadratic degree growth and is integrable in the sense of vanishing entropy. 
        Since the two copies of $\widetilde{W}(A_1^{(1)})$ commute, the systems \eqref{eq:4DqP12} and \eqref{eq:4DqP1} are symmetries of each other.
    \end{remark}

\section{Conclusion}

The main motivation of this paper is the expected general geometric theory of discrete Painlev\'e equations in higher dimensions, extending the celebrated Sakai framework.
To this end, the purpose of this particular study is to add to the collection of examples of higher-order analogues of discrete Painlev\'e equations for which the geometry of their spaces of initial conditions has been worked out.

We began with an autonomous fourth-order discrete system which naturally extends the autonomous degenerate case of a $q$-discrete Painlev\'e equation of surface type $A_7^{(1)}$ in the Sakai classification.
This is defined by a birational self-map of $(\p^1)^{\times 4}$ and constructed its space of initial conditions, a rational variety of which the map becomes a pseudo-automorphism, obtained by 28 blow-ups. 
The geometry of the resulting variety $X$ was used to find conserved quantities as well as establishing that it exhibits quadratic degree growth. 

After making a choice of effective anticanonical divisor of $X$, we embedded $X$ into a family of varieties leaving this intact, which admits an action of $\widetilde{W}(A_1^{(1)}) \times \widetilde{W}(A_1^{(1)})$ by pseudo-isomorphims, described in terms of  root system structures in the Neron-Severi bi-lattice of type $(\underset{ \left< \alpha, \alpha^\vee \right> = 14}{A_1 }+ A_1)^{(1)}$. 
From this we constructed two fourth-order analogues of $q$-discrete Painlev\'e equations from the actions of translation elements of $\widetilde{W}(A_1^{(1)}) \times \widetilde{W}(A_1^{(1)})$, one of which is a deautonomisation of the autonomous fourth-order system we started with.
It would be interesting to compare these with the higher-order analogues of $q$-discrete Painlev\'e equations of surface types $A_6^{(1)}$ and $A_7^{(1)}$ constructed in \cite{Masuda-Okubo-Tsuda-2021}.

However, we are at this stage still far from a general theory of higher-order discrete Painlev\'e equations in terms of rational varieties. 
We need more examples (in particular ones that are not as directly connected to second-order discrete Painlev\'e equations) as well as inputs from other perspectives on higher discrete Painlev\'e equations from related topics, e.g. cluster algebras, complex dynamics, higher-dimensional birational geometry,
before a general theory based on an analogue of a generalised Halphen surface in higher dimensions can be achieved.
\subsection*{Acknowledgements}

AS was supported by the Japan Society for the Promotion of Science (JSPS) through KAKENHI Grant-in-Aid (24K22843). 
TT was supported by JSPS through KAKENHI Grand-in-Aid (C) (22K03383).



\appendix 
\section{Order of blow-ups and divisor class groups}\label{order_blow-ups}

In our study of the singular locus near $q_1=0$, the inclusion relations as indicated in Figure \ref{fig:firstsevenblow-ups} were observed, and we performed blow-ups corresponding to the divisors $V_4$, $V_5$, and $V_6$ after those corresponding to $V_1$ and $V_2$. A natural question arises: \emph{Is it possible to perform the blow-ups of $V_4$, $V_5$, and $V_6$ prior to those of $V_1$ and $V_2$?} The answer is affirmative, as we explain here.

To elaborate, consider a birational map 
\[
f: X \dashrightarrow Y,
\]
between smooth algebraic varieties, and let 
\[
\pi: \tilde{X} \to X \quad \text{and} \quad \pi': \tilde{X}' \to X
\]
be sequences of blow-ups chosen so that the induced maps 
\[
\tilde{f} = f\circ \pi: \tilde{X} \dashrightarrow Y \quad \text{and} \quad \tilde{f}' = f\circ \pi': \tilde{X}' \dashrightarrow Y
\]
resolve the indeterminacies of $f$ in codimension one. If the generic points of the exceptional divisors corresponding to $\pi$ and $\pi'$ agree, then the hypotheses of the following well-known proposition are met by $g: \tilde{X} \dashrightarrow \tilde{X}'$ that acts trivially on generic points. In particular, the divisor class groups $\operatorname{Cl}(\tilde{X})$ and $\operatorname{Cl}(\tilde{X}')$ are isomorphic. In the case of $X$ being rational, since the N\'eron-Severi bilattice is determined solely by the divisor class group, it follows that the corresponding N\'eron-Severi bilattices are isomorphic as well.
\begin{proposition}[Isomorphism of Divisor Class Groups]\label{prop:iso_divisor}
Let $X$ and $Y$ be normal varieties and let
\[
g: X \dashrightarrow Y
\]
be a birational map. Suppose there exist open subsets $U\subset X$ and $V\subset Y$ such that
\[
\operatorname{codim}(X\setminus U)\ge2 \quad \text{and} \quad \operatorname{codim}(Y\setminus V)\ge2,
\]
and such that the restriction
\[
g|_U: U \xrightarrow{\sim} V
\]
is an isomorphism. Then the induced homomorphism
\[
g_*: \operatorname{Cl}(X)\to\operatorname{Cl}(Y)
\]
between the divisor class groups is an isomorphism.
\end{proposition}


\begin{proof}
Let $g: X \dashrightarrow Y$ be a birational map that restricts to an isomorphism 
\[
g|_U: U\to V
\]
on open subsets $U\subset X$ and $V\subset Y$ with $\operatorname{codim}(X\setminus U)\ge2$ and $\operatorname{codim}(Y\setminus V)\ge2$. Since every prime divisor in $X$ has its generic point contained in $U$, we can define a pushforward map
\[
g_*:\operatorname{Div}(X)\to\operatorname{Div}(Y)
\]
by sending a prime divisor $E\subset X$ (with $E\cap U\neq\varnothing$) to the closure in $Y$ of $g(E\cap U)$. For any rational function $\phi\in K(Y)$, the isomorphism $g|_U$ ensures that
\[
\operatorname{ord}_E\bigl(g^*\phi\bigr)=\operatorname{ord}_{g_*(E)}(\phi).
\]
It follows that
\[
g_*\bigl(\operatorname{div}_X(g^*\phi)\bigr)=\operatorname{div}_Y(\phi),
\]
so that $g_*$ sends principal divisors to principal divisors. Consequently, $g_*$ induces a well-defined homomorphism
\[
g_*:\operatorname{Cl}(X)\to\operatorname{Cl}(Y).
\]
Injectivity follows from the fact that if $g_*(D)$ is trivial in $\operatorname{Cl}(Y)$, then $D$ must be principal. Surjectivity is obtained by applying the same argument to the inverse map $g^{-1}$ on $V$. 
\end{proof}

In our specific case, we consider the birational map
\[
f: \bigl(\mathbb{P}^1\bigr)^{\times 4} \dashrightarrow \bigl(\mathbb{P}^1\bigr)^{\times 4}.
\]
Let $E_i$ for $i=1,\dots,6$ denote the total transforms of the exceptional divisors obtained by performing the blow-ups in the order $1,2,3,4,5,6$, and let $E_i'$ denote the corresponding total transforms when the blow-ups are performed in the order $4,5,6,1,2,3$. In the former case, the only inclusion in the base locus is
\[
V_3 \subset E_2,
\]
so that the prime divisors in the exceptional part of $\tilde{X}$ are given by $E_1$, $E_3$, $E_4$, $E_5$, $E_6$, and the difference $E_2 - E_3$. In contrast, in the latter case the base locus exhibits the inclusion relations
\[
V_1 \subset E_4', \quad V_2 \subset E_5', \quad \text{and} \quad V_3 \subset E_2',
\]
so that the prime divisors in the exceptional part of $\tilde{X}'$ are $E_1'$, $E_3'$, $E_6'$ and the differences $E_4'-E_1'$, $E_5'-E_2'$, and $E_2'-E_3'$.

The anticanonical divisors also correspond through a base change as
\[
-K_{\tilde{X}} = -\pi^*\Bigl(K_{(\mathbb{P}^1)^{\times 4}}\Bigr)
- 2 E_1 - 2 E_2
- E_3 -E_4-E_5-E_6,
\]
and
\begin{align*}
-K_{\tilde{X}'} &= -(\pi')^*\Bigl(K_{(\mathbb{P}^1)^{\times 4}}\Bigr)
-2 E_1' -2\,E _2' -E_3'
 -(E_4'-E_1')-(E_5'-E_2')-E_6'\\
&=- (\pi')^*\Bigl(K_{(\mathbb{P}^1)^{\times 4}}\Bigr)
-\sum_{i=1}^6 E_i',
\end{align*}
where
\[
-\pi^*\Bigl(K_{(\mathbb{P}^1)^{\times 4}}\Bigr)=2H_{q_1}+2H_{q_2}+2H_{p_1}+2H_{p_2}.
\]


\begin{thebibliography}{100}

\bibitem{Alonso-Suris-Wei}
Alonso, J., Suris, Y. B., Wei, K. 
\emph{A three-dimensional generalization of QRT maps}, 
J. Nonlinear Sci. 33 (2023), no. 6, Paper No. 117, 27 pp.

\bibitem{Joshi-Alrashdi}
Alrashdi, H., Joshi, N., Tran, D. T. 
\emph{Hierarchies of q-discrete Painlev\'e equations}, 
J. Nonlinear Math. Phys. 27 (2020), no. 3, 453–477.

\bibitem{bayraktar}
Bayraktar, T.
\emph{Green currents for meromorphic maps of compact K\"ahler manifolds}, 
J. Geom. Anal. 23 (2013), no. 2, 970–998.

\bibitem{Bedford-Kim-2004}
Bedford, E., Kim, K. 
\emph{On the degree growth of birational mappings in higher dimension},
J. Geom. Anal. 14 (2004), no. 4, 567–596.

\bibitem{Bedford-Kim-2008}
Bedford, E., Kim, K. 
\emph{Degree growth of matrix inversion: birational maps of symmetric, cyclic matrices},
Discrete Contin. Dyn. Syst. 21 (2008), no. 4, 977–1013.

\bibitem{Bell-Diller-Jonsson-Krieger-2024}
Bell, J. P., Diller, J., Jonsson, M., Krieger, H.,
\emph{Birational maps with transcendental dynamical degree}, 
Proc. Lond. Math. Soc. (3) 128 (2024), no. 1, Paper No. e12573, 47 pp.

\bibitem{Bellon-Viallet}
Bellon, M. P., Viallet, C.-M.,
\emph{Algebraic entropy}
Comm. Math. Phys. 204 (1999), no. 2, 425--437.

 \bibitem{stefan} 
 Carstea, A. S.,  Takenawa, T.,
\emph{Space of initial conditions and geometry of two 4-dimensional discrete Painlev\'e equations},
J. Phys. A 52 (2019), no. 27, 275201, 25 pp. 

\bibitem{algebraicallystablevariety}
 Carstea, A. S.,  Takenawa, T.,
\emph{An algebraically stable variety for a four-dimensional dynamical system reduced from the lattice super-KdV equation}, 
Springer Proc. Math. Stat. 338, 43–53.
Springer, Cham, 2020


\bibitem{fane1}
Carstea, A. S.,  Dzhamay, A., Takenawa, T.
\emph{Fiber-dependent deautonomization of integrable 2D mappings and discrete Painlev\'e equations} 
J. Phys. A 50 (2017), no. 40, 405202, 41 pp.

\bibitem{Joshi-Cresswell}
Cresswell, C., Joshi, N. 
\emph{The discrete first, second and thirty-fourth Painlevé hierarchies}, 
J. Phys. A 32 (1999), no. 4, 655–669.

\bibitem{dillerfavre}
Diller, J., Favre, C.
\emph{Dynamics of bimeromorphic maps of surfaces},
Amer. J. Math. 123 (2001), no. 6, 1135–1169.

\bibitem{dolgachevweylgroups}
Dolgachev, I.
\emph{Weyl groups and Cremona transformations},
Singularities, Part 1 (Arcata, Calif., 1981), 283–294.
Proc. Sympos. Pure Math., 40
American Mathematical Society, Providence, RI, 1983


\bibitem{Dolgachev-Ortland-1989}
Dolgachev, I., Ortland, D., 
\emph{Point sets in projective spaces and theta functions}, 
Ast\'erisque No. 165 (1988), 210 pp.


\bibitem{dP2}
Dzhamay, A., Shi, Y., Stokes, A., Willox, R.,
\emph{What is the symmetry group of a d-}$\pain{II}$ \emph{discrete Painlev\'e equation?}, 
Mathematics 13 (2025), no. 7, 1123.

\bibitem{FS1995} 
Fornaess,~J.~E.,  Sibony,~N., 
\emph{Complex dynamics in higher dimension. II},
in Bloom, T., Catlin, D., D'Angelo, J. P., Siu, Y.-T. eds. \emph{Modern methods in complex analysis (Princeton, NJ, 1992)}, 135–182.
Ann. of Math. Stud., 137
Princeton University Press, Princeton, NJ, 1995.

\bibitem{GRP-1991}
Grammaticos, B., Ramani, A., Papageorgiou, V.,
\emph{Do integrable mappings have the Painlev\'e property?},
Phys. Rev. Lett. 67 (1991), no. 14, 1825--1828.

\bibitem{hietarintaviallet}
Hietarina, J., Viallet, C.-M., 
\emph{Singularity confinement and chaos in discrete systems}
Phys. Rev. Lett. 81 (1998), 325--328. 


\bibitem{projectivereduction}
Kajiwara, K., Nakazono, N., Tsuda, T.,
\emph{Projective reduction of the discrete Painlevé system of type} $(A_2+A_1)^{(1)}$.
Int. Math. Res. Not. IMRN 2011, no. 4, 930–966.

\bibitem{Kawakami-I}
Kawakami, H.,
\emph{Four-dimensional Painlev\'e-type equations associated with ramified linear equations I: Matrix Painlevé systems},
Funkcial. Ekvac. 63 (2020), no. 1, 97–132.

\bibitem{Kawakami-II}
Kawakami, H.,
\emph{Four-dimensional Painlev\'e-type equations associated with ramified linear equations II: Sasano systems},
J. Integrable Syst. 3 (2018), no. 1, xyy013, 36 pp.


\bibitem{Kawakami-III}
Kawakami, H.,
\emph{Four-dimensional Painlev\'e-type equations associated with ramified linear equations III: Garnier systems and Fuji-Suzuki systems},
SIGMA Symmetry Integrability Geom. Methods Appl. 13 (2017), Paper No. 096, 50 pp.

\bibitem{27}
Kawakami, H., Nakamura, A., Sakai, H., 
\emph{Degeneration scheme of 4-dimensional Painlevé-type equations},
MSJ Mem. 37,  25–111.
Mathematical Society of Japan, Tokyo, 2018

\bibitem{Mase}
Mase, T.,
\emph{Studies on spaces of initial conditions for non-autonomous mappings of the plane},
J. Integrable Syst. 3 (2018), no. 1, xyy010, 47 pp.

\bibitem{Masuda-2015}
Masuda, T.,
\emph{A }$q$\emph{-analogue of the higher order Painlev\'e type equations with the affine Weyl group symmetry of type }$D$, 
Funkcial. Ekvac. 58 (2015), no. 3, 405–430.

\bibitem{Masuda-Okubo-Tsuda-2021}
Masuda, T., Okubo, N., Tsuda, T.,
\emph{Cluster algebras and higher order generalizations of the }$q$\emph{-Painlev\'e equations of type }$A_7^{(1)}$\emph{ and }$A_6^{(1)}$,
RIMS Kôkyûroku Bessatsu B87, 149–163, 2021.


\bibitem{Nakazono-2023}

Nakazono, N.,
\emph{Higher-order generalizations of the }${A}_{6}^{(1)}$\emph{- and }${A}_{4}^{(1)}$\emph{-surface type q-Painlevé equations},
Phys. Scr. 98 (2023), no.~11, 115204.

\bibitem{Nakazono-2024}
Nakazono, N.,
\emph{A higher-order generalization of an }${A}_{4}^{(1)}$\emph{-surface type q-Painlevé equation with }$\widetilde{W}({({A}_{2N}\rtimes {A}_{1})}^{(1)}\times {A}_{1}^{(1)})$\emph{ symmetry},
Phys. Scr. 99 (2024), no.~8, 085214.

\bibitem{OKAMOTO1979} 
Okamoto, K., 
\emph{Sur les feuilletages associ\'es aux \'equations du second ordre \`a points critiques fixes de P. Painlev\'e}, 
Japan. J. Math. (N.S.) 5 (1979), no. 1, 1--79. 

\bibitem{Okubo-Suzuki-2020}
Okubo, N., Suzuki, T.,
\emph{Cluster algebra and }$q$\emph{-Painlev\'e equations: higher order generalization and degeneration structure},
RIMS Kôkyûroku Bessatsu B78, 53–75, 2020.

\bibitem{Okubo-Suzuki-2022}
Okubo, N., Suzuki, T.,
\emph{Generalized }$q$\emph{-Painlevé VI systems of type }$(A_{2n+1}+A_1+A_1)^{(1)}$ \emph{ arising from cluster algebra},
Int. Math. Res. Not. IMRN 2022, no. 9, 6561–6607.


\bibitem{roeder}
Roeder, R. K. W.
\emph{The action on cohomology by compositions of rational maps}, 
Math. Res. Lett. 22 (2015), no. 2, 605–632.




\bibitem{sakai} 
Sakai, H., 
\emph{Rational surfaces associated with affine root systems and geometry of the Painlev\'e equations}, 
Comm. Math. Phys. 220 (2001), no. 1, 165–229.

\bibitem{37}
Sakai, H., 
\emph{Isomonodromic deformation and 4-dimensional Painlev\'e-type equations}, 
MSJ Mem. 37, 1–23.
Mathematical Society of Japan, Tokyo, 2018. 

\bibitem{yangtranslations}
Shi, Y.,
\emph{Translation in affine Weyl groups and its application in discrete integrable systems}
Proc. R. Soc. A. 481 (2025), 20240749.

\bibitem{Suzuki-2015}
Suzuki, T.,
\emph{A }$q$\emph{-analogue of the Drinfeld-Sokolov hierarchy of type }$A$\emph{ and }$q$\emph{-Painlev\'e system},
Contemp. Math. 651, 25–38.
American Mathematical Society, Providence, RI, 2015.

\bibitem{Suzuki-2017}
Suzuki, T.,
\emph{A reformulation of the generalized }$q$\emph{-Painlev\'e VI system with }$W(A_{2n+1}^{(1)})$ \emph{ symmetry}, 
J. Integrable Syst. 2 (2017), no. 1, xyw017, 18 pp.

\bibitem{Takenawa-2004}
Takenawa, T., 
\emph{Discrete Dynamical Systems Associated with the Configuration Space of 8 Points in }$\p^3(\C)$, 
Commun. Math. Phys. 246 (2004), 19--42.

\bibitem{Takenawa-2021} 
Takenawa, T. 
\emph{Space of initial conditions for the four dimensional Fuji-Suzuki-Tsuda system}, RIMS Kôkyûroku Bessatsu B87, 99-111, 2021.

\bibitem{Takenawa-2024}
Takenawa, T. 
\emph{Space of initial conditions for the four-dimensional Garnier system revisited},
RIMS Kôkyûroku Bessatsu B96, 117-130, 2024.

\bibitem{Tsuda-2004}
Tsuda, T.,
\emph{Integrable mappings via rational elliptic surfaces},
J. Phys. A 37 (2004), no. 7, 2721–2730.

\bibitem{Tsuda-2010}
Tsuda, T.,
\emph{On an integrable system of q-difference equations satisfied by the universal characters: its Lax formalism and an application to }$q$\emph{-Painlev\'e equations}, 
Comm. Math. Phys. 293 (2010), no. 2, 347–359.

\bibitem{Tsuda-Takenawa-2009}
Tsuda, T., Takenawa, T., 
\emph{Tropical representation of Weyl groups associated with certain rational varieties}, 
Adv. Math. 221 (2009), no. 3, 936–954.



\end{thebibliography}
\end{document}